\documentclass[11pt]{article}
\usepackage{amssymb}

\usepackage{graphicx}
\usepackage{amsmath}
\usepackage{makeidx}
\usepackage{indentfirst}

\newcounter{resultnum}[section]
\newtheorem{conclusion}{Conclusion}[section]

\newcounter{conclusionnum}[section]
\newtheorem{condition}{Condition}[section]

\newcounter{conditionnum}[section]

\newcounter{conjecturenum}[section]

\newcounter{examplenum}[section]

\newcounter{exercisenum}[section]
\newtheorem{lemma}{Lemma}[section]

\newcounter{lemmanum}[section]

\newcounter{notationnum}[section]
\newtheorem{theorem}{Theorem}[section]

\newcounter{theoremnum}[section]
\newtheorem{definition}{Definition}[section]

\newcounter{definitionnum}[section]
\newtheorem{corollary}{Corollary}[section]

\newcounter{corollarynum}[section]
\newtheorem{remark}{Remark}[section]

\newcounter{remarknum}[section]
\newtheorem{proposition}{Proposition}[section]

\newcounter{propositionnum}[section]

\newcounter{acknowledgementnum}[section]

\newcounter{algorithmnum}[section]

\newcounter{axiomnum}[section]

\newcounter{casenum}[section]

\newcounter{claimnum}[section]

\newcounter{summarynum}[section]

\newcounter{problemnum}[section]
\newenvironment{proof}[1][]{\textbf{Proof.} }{}

\begin{document}

\title{On General Solutions for Field Equations in Einstein and Higher Dimension Gravity}
\date{December 2, 2009}
\author{ Sergiu I. Vacaru\thanks{
sergiu.vacaru@uaic.ro, Sergiu.Vacaru@gmail.com;\newline
http://www.scribd.com/people/view/1455460-sergiu } \\
{\quad} \\
{\small {\textsl{\ Science Department, University "Al. I. Cuza" Ia\c si},} }%
\\
{\small {\textsl{\ 54 Lascar Catargi street, 700107, Ia\c si, Romania}} }}
\maketitle

\begin{abstract}
We prove that the Einstein equations can be solved in a very general form for arbitrary spacetime dimensions and various types of vacuum and non--vacuum cases following a geometric method of anholonomic frame deformations for constructing exact solutions in gravity. The main idea of this method is to introduce on (pseudo) Riemannian manifolds an alternative (to the Levi--Civita connection) metric compatible linear connection which is also completely defined by the same metric structure. Such a canonically distinguished connection is with nontrivial torsion which is induced by some nonholonomy frame coefficients and generic off--diagonal terms of metrics. It is possible to define certain classes of adapted frames of reference when  the Einstein equations for such an  alternative connection transform into a system of partial differential equations which can be integrated in very general forms. Imposing nonholonomic constraints on generalized metrics and connections and adapted frames (selecting Levi--Civita configurations), we generate exact solutions in Einstein gravity and extra dimension generalizations.

\vskip0.1cm


\vskip3pt


\end{abstract}


\section{Introduction and formulation of Main Result}

The issue to construct exact solutions in Einstein gravity and high
dimensional gravity theories is not new. It has been posed in different ways
and related to various problems in multidimensional cosmology, black hole
physics, nonlinear gravitational effects etc in Kaluza--Klein gravity and
string/brane generalizations. For instance, in brane gravity, one faces the
problem to generate solutions with wrapped configurations and possible
quantum corrections and noncommutative modifications. Former elaborated
approaches are model dependent, usually for metric ansatz depending on 1-2
coordinates, for spherical/cylindrical symmetries and chosen backgrounds
with certain types of asymptotic boundary conditions. In this context, and
further application in modern physics, it is important to find a way to
elaborate methods of constructed exact solutions in very general form.

The problem of constructing most general classes of solutions in gravity has
been posed in a geometric language following the anholonomic deformation/
frame method, see reviews of results in Refs. \cite{ijgmmp,vncg,vsgg,vrflg}.%
\footnote{%
The geometry of nonholonomic distributions/ deformations and frames should
be not identified with the Cartan's moving frame method even in the first
case "moving frames" can be also included. In our approach, we consider
arbitrary real/complex, in general, noncommutative/supersymmetric
nonholonomic distributions on certain manifolds and adapt the geometric
constructions with respect to such distributions. Such constructions result
in (nonlinear) deformations of connection and metric structures, which is
not the case for moving frames, when the same geometric objects are
re--expressed with respect to moving/different systems of reference.
Selecting some convenient nonholonomic distributions, we obtain separations
of equations and reparametrizations of variables in some physically
important nonlinear systems of partial differential equations which allows
us to integrate such systems in general forms. Then constraining
correspondingly some general solutions, we select necessary subclasses of
exact solutions, for instance, in general relativity and extra dimensions.}
In brief, the method allows us to generate any spacetime metric with
prescribed, or reasonable, physical/ geometric properties (as solutions of
Einstein equations and modifications for different gravity theories) by
performing certain types of nonholonomic transforms/ deformations from
another well defined (pseudo) Riemannian metrics. Such a problem of
constructing "almost general" solutions for Einstein spaces was solved recently for
four and five dimensional spaces \cite{vgseg} but straightforward extensions
and more sophisticate constructions should be provided to achieve such
results for spacetimes of arbitrary dimensions.\footnote{%
Some conditions of theorems and formulas, for four and five dimensions,
presented in Ref. \cite{vgseg} (that version should be considered as a
Letter version of this work, where we emphasized certain techniques for
generating solutions in general relativity) will be repeated for some our
further constructions because they are used for different type
generalizations and simplify proofs for higher dimensions.}

In this paper, we show how the Einstein equations can be integrated in
general form for spacetimes of higher dimensions $\ \ ^{k}n=n+\ m+\
^{1}m+...+\ ^{k}m>5,$ when $n=2,$ or $3,$ and $\ ^{0}m=m,\ ^{1}m,\
^{2}m,...=2;$ for $k=0,1,2,... .$ Such constructions provide a
generalization of the Main Result (Theorem 1) from Ref. \cite{vgseg} proved
for dimensions $\ ^{0}n=n+\ m=4,$ or $5$ (see details on geometric methods
of constructing exact solutions in gravity in Refs. \cite%
{ijgmmp,vncg,vsgg,vrflg}). The approach developed in this work may present a
substantial interest for research in higher dimensional (super) gravity
theories (which, in general, may possess higher order anisotropies \cite%
{vnp1,vsp1,vv1,vst}) and in higher order Lagrange--Finsler/ Hamilton--Cartan
geometry and related gravity and mechanical models \cite%
{m1,m2,m3,av,ma1987,ma}.

Let us consider a (pseudo) Riemannian manifold $\ ^{k}\mathbf{V,}$ $dim\ ^{k}%
\mathbf{V=\ }\ ^{k}n,$ provided with a metric
\begin{equation}
\mathbf{g}=\mathbf{g}_{\ ^{k}\alpha \ ^{k}\beta }(u^{\ ^{k}\gamma })du^{\
^{k}\alpha }\otimes du^{\ ^{k}\beta }  \label{metrml}
\end{equation}%
of arbitrary signature $\epsilon _{\ ^{k}\alpha }=(\epsilon (1)=\pm
1,\epsilon (2)=\pm 1,\ldots ,\epsilon (\ ^{k}n)\pm 1).$\footnote{%
\textbf{Notation Remarks:} In our works, we follow conventions from \cite%
{ijgmmp,vrflg,vnp1,vsp1,vst} when left up/low indices are used as labels for
certain types of geometric spaces/manifolds and objects. The Einstein
summation rule is applied on repeating right left--up indices if it is not
stated a contrary condition. Boldfaced letters are used for spaces
(geometric objects) enabled with (adapted to) some nonholonomic
distributions/frames prescribed on corresponding classes of manifolds. An
abstract/coordinate index formalism is necessary for deriving in explicit
form some general/ exact solutions for gravitational field equations in
higher dimensional gravity. Unfortunately, such denotations can not be
introduced in a more simple form if we aim to present certain general
results on exact solutions derived for some ''multi--level'' systems of
nonlinear partial differential equations.} The local coordinates on $\ ^k%
\mathbf{V}$ are parametrized ''shell by shell'' by increasing dimensions on $%
2$ at every level (equivalently, shell). We begin with denotations for $\
^{0}\mathbf{V,}$ $dim\ ^{0}\mathbf{V=\ }\ ^{0}n,$ when $u^{\alpha
}=(x^{i},y^{a}),$ for $u^{\alpha }=u^{^{0}\alpha }$ and $y^{a}=y^{\ ^{0}a}$
with $k=0,$ where $x^{i}=(x^{1},x^{\widehat{i}})$ and $y^{a}=\left(
v,y\right) ,$ i. e. $y^{4}=v,$ $y^{5}=y.$ Indices $i,j,k,...=1,2,3;$ $\hat{%
\imath},\hat{\jmath},\hat{k}...=2,3$ and $a,b,c,...=4,5$ are used for a
conventional $(3+2)$--splitting of dimension and general abstract/coordinate
indices when $\alpha ,\beta ,\ldots $ run values $1,2,...,5.$ For four
dimensional (in brief, 4--d) constructions, we can write $u^{\widehat{\alpha
}}=(x^{\widehat{i}},y^{a}),$ when the coordinate $x^{1}$ and values for
indices like $\alpha ,i,...=1$ are not considered. In brief, we shall denote
some partial derivatives $\partial _{\alpha }=\partial /\partial u^{\alpha }$
in the form $s^{\bullet }=\partial s/\partial x^{2},s^{\prime }=\partial
s/\partial x^{3},s^{\ast }=\partial s/\partial y^{4}.$ At the next level, $\
^{1}\mathbf{V,}$ $dim\mathbf{\ }\ ^{1}\mathbf{V=\ }\ ^{1}n=n+\ m+\ ^{1}m,$
the coordinates are labeled $u^{^{1}\alpha }=(x^{i},y^{a},y^{^{1}a}),$ for $%
^{1}a,\ ^{1}b,...=6,7.$ For the ''$2$--anisotropy'', $\mathbf{\ }\ ^{2}%
\mathbf{V,}$ $dim\mathbf{\ }\ ^{2}\mathbf{V=\ }\ ^{2}n=n+\ m+\ ^{1}m+\
^{2}m, $ the coordinates are labeled $u^{^{2}\alpha
}=(x^{i},y^{a},y^{^{1}a},y^{^{2}a}),$ for $\ ^{2}a,\ ^{2}b,...=8,9;$ and
(recurrently) for the ''$k$--anisotropy'', $\ ^{k}\mathbf{V,}$ $dim\mathbf{\
}\ ^{k}\mathbf{V=\ }\ ^{k}n,$ the coordinates are labeled $u^{^{k}\alpha
}=(x^{i},y^{a},y^{^{1}a},y^{^{2}a},...,y^{^{k}a}),$ for $\ ^{k}a,\
^{k}b,...=4+2k,5+2k.$\footnote{%
we use the term anisotropy/ anisotropic for some nonholonomically
(equivalently, anholonomically) constrained variables/ coordinates on a
(pseudo) Riemannian manifold subjected to certain non--integrable
conditions; such anisotropies should be not confused with those when
geometric objects depend on some ''directions and velocities'', for
instance, in Finsler geometry}

We shall write $\ ^{k}\nabla =\{\Gamma _{\ \ ^{k}\beta \ ^{k}\gamma }^{\
^{k}\alpha }\}$ for the Levi--Civita connection, with coefficients stated
with respect to an arbitrary local frame basis $e_{\ ^{k}\alpha }=(e_{\
^{k-1}\alpha },$ $e_{\ ^{k}a})$ and its dual basis $e^{\ ^{k}\beta }=(e^{\
^{k-1}\beta },e^{\ ^{k}b}).$ Using the Riemannian curvature tensor $\ ^{k}%
\mathcal{R}=\{R_{\ \ ^{k}\beta \ ^{k}\gamma \ ^{k}\delta }^{\ ^{k}\alpha }\}$
defined by $\ ^{k}\nabla ,$ one constructs the Ricci tensor, $\ ^{k}\mathcal{%
R}ic=\{R_{\ \ ^{k}\beta \ ^{k}\delta }\doteqdot R_{\ \ ^{k}\beta \
^{k}\alpha \ ^{k}\delta }^{\ ^{k}\alpha }\},$ and scalar curvature $\
^{k}R\doteqdot \mathbf{g}^{\ ^{k}\beta \ ^{k}\delta }R_{\ \ ^{k}\beta \
^{k}\delta },$ where $\mathbf{g}^{\ ^{k}\beta \ ^{k}\delta }$ is inverse to $%
\mathbf{g}_{\ ^{k}\alpha \ ^{k}\beta }.$ The Einstein equations on $\ ^{k}%
\mathbf{V,}$ for an energy--momentum source $\ ^{k}T_{\alpha \beta },$ are
written in the form%
\begin{equation}
R_{\ \ ^{k}\beta \ ^{k}\delta }-\frac{1}{2}\mathbf{g}_{\ ^{k}\beta \
^{k}\delta }\ ^{k}R=\varkappa T_{\ ^{k}\beta \ ^{k}\delta },  \label{einsteq}
\end{equation}%
where $\varkappa =const.$ For the Einstein spaces defined by a cosmological
constant $\lambda ,$ such gravitational field equations can be represented
as $R_{\ \ \ ^{k}\beta }^{\ ^{k}\alpha }=\lambda \delta _{\ ^{k}\beta }^{\
^{k}\alpha },$ where $\delta _{\ ^{k}\beta }^{\ ^{k}\alpha }$ is the
Kronecher symbol. The vacuum solutions are obtained for $\lambda =0.$

The goal of our work is to provide (see further sections) the proof of:

\begin{theorem}
\label{mth}\textbf{(Main Theorem)} The gravitational field equations in the $%
\ ^{k}n$--dimensional Einstein gravity (\ref{einsteq}) represented by frame
transforms as
\begin{equation}
R_{\ \ \ ^{k}\beta }^{\ ^{k}\alpha }=\Upsilon _{\ \ \ ^{k}\beta }^{\
^{k}\alpha },  \label{einst1}
\end{equation}%
for any given $\Upsilon _{\ \ \ ^{k}\beta }^{\ ^{k}\alpha }=diag[\Upsilon
_{\ 1},\Upsilon _{\ 2},\Upsilon _{\ 2}=\Upsilon _{\ 3},\Upsilon
_{4},\Upsilon _{5}=\Upsilon _{4},\Upsilon _{6}=\ ^{1}\Upsilon _{2},\newline
\Upsilon _{7}=\Upsilon _{6},\ldots ,\Upsilon _{4+2k}=\ ^{k}\Upsilon
_{2},\Upsilon _{5+2k}=\Upsilon _{4+2k}]$ with
\begin{eqnarray}
\Upsilon _{\ \ 1}^{1} &=&\Upsilon _{1}=\Upsilon _{2}+\Upsilon _{4},\
\Upsilon _{\ \ 2}^{2}=\Upsilon _{\ \ 3}^{3}=\Upsilon _{2}(x^{k},v),
\label{source} \\
\Upsilon _{\ \ 4}^{4} &=&\Upsilon _{\ \ 5}^{5}=\Upsilon _{4}(x^{\widehat{k}%
}),\Upsilon _{\ \ 6}^{6}=\Upsilon _{\ \ 7}^{7}=\ ^{1}\Upsilon _{2}(u^{\
\alpha },\ ^{1}v),\ \Upsilon _{\ \ 8}^{8}=  \notag \\
\Upsilon _{\ \ 9}^{9} &=&\ ^{2}\Upsilon _{2}(u^{\ ^{1}\alpha },\
^{2}v),\ldots ,\Upsilon _{\ \ 4+2k}^{4+2k}=\Upsilon _{\ \ 5+2k}^{5+2k}=\
^{k}\Upsilon _{2}(u^{\ ^{k-1}\alpha },\ ^{k}v),  \notag
\end{eqnarray}%
for $\ y^{4}=v,$ $y^{6}=\ ^{1}v,y^{8}=\ ^{2}v,...,y^{4+2k}=\ ^{k}v$ (where $k
$ labels the shell's number), can be solved in general form by metrics of
type
\begin{eqnarray}
\ ^{k}\mathbf{g} &\mathbf{=}&\epsilon _{1}{dx^{1}\otimes dx^{1}}+g_{\widehat{%
i}}(x^{\widehat{k}}){dx^{\widehat{i}}\otimes dx^{\widehat{i}}}+\omega
^{2}(x^{j},y^{b})h_{a}(x^{k},v)\mathbf{e}^{a}{\otimes }\mathbf{e}^{a}  \notag
\\
&&+\ ^{1}\omega ^{2}(u^{\alpha },y^{\ ^{1}b})h_{\ ^{1}a}(u^{\alpha },\
^{1}v)\ \mathbf{e}^{\ ^{1}a}{\otimes }\ \mathbf{e}^{\ ^{1}a}
\label{ansgensol} \\
&&+\ ^{2}\omega ^{2}(u^{\ ^{1}\alpha },y^{\ ^{2}b})h_{\ ^{2}a}(u^{\
^{1}\alpha },\ ^{2}v)\ \mathbf{e}^{\ ^{2}a}{\otimes }\ \mathbf{e}^{\
^{2}a}+...  \notag \\
&&+\ ^{k}\omega ^{2}(u^{\ ^{k-1}\alpha },y^{\ ^{k}b})h_{\ ^{2}a}(u^{\
^{k-1}\alpha },\ ^{k}v)\ \mathbf{e}^{\ ^{k}a}{\otimes }\ \mathbf{e}^{\
^{k}a},  \notag
\end{eqnarray}
\begin{eqnarray}
\mbox{ for }\mathbf{e}^{4} &=&dy^{4}+w_{i}(x^{k},v)dx^{i},\mathbf{e}%
^{5}=dy^{5}+n_{i}(x^{k},v)dx^{i},  \notag \\
\mathbf{e}^{6} &=&dy^{6}+w_{\beta }(u^{\alpha },\ ^{1}v)du^{\beta },\mathbf{e%
}^{7}=dy^{7}+n_{\beta }(u^{\alpha },\ ^{1}v)du^{\beta },  \notag \\
\mathbf{e}^{8} &=&dy^{8}+w_{\ ^{1}\beta }(u^{\ ^{1}\alpha },\ ^{2}v)du^{\
^{1}\beta },\mathbf{e}^{9}=dy^{9}+n_{\ ^{1}\beta }(u^{\ ^{1}\alpha },\
^{2}v)du^{\ ^{1}\beta },  \notag \\
&&......  \notag \\
\mathbf{e}^{4+2k} &=&dy^{4+2k}+w_{\ ^{k-1}\beta }(u^{\ ^{k-1}\alpha },\
^{k}v)du^{\ ^{k-1}\beta },  \notag \\
\mathbf{e}^{5+2k} &=&dy^{5+2k}+n_{\ ^{k-1}\beta }(u^{\ ^{k-1}\alpha },\
^{k}v)du^{\ ^{k-1}\beta },  \notag
\end{eqnarray}%
where coefficients are defined by generating functions $f(x^{i},v),\partial
f/\partial v\neq 0,$ $...$ $,$ $\ ^{k}f(u^{\ ^{k-1}\alpha },\ ^{k}v),$ $%
\partial \ ^{k}f/\partial \ ^{k}v\neq 0$ and $\omega (x^{j},y^{b}),...,\
^{k}\omega (u^{\ ^{k-1}\alpha },\ y^{\ ^{k}b})\neq 0$ and integration
functions $^{0}f(x^{i}),$ $...,$ $\ _{k}^{0}f(u^{\ ^{k-1}\alpha }),$ $\
^{0}h(x^{i}),$ $...,$ $_{k}^{0}h(u^{\ ^{k-1}\alpha }),$ $\
^{1}n_{j}(x^{i}),...,\ \ n_{\ ^{k-1}\beta }(u^{\ ^{k-1}\alpha }),$ $\
^{2}n_{j}(x^{i}),...,\ _{^{k-1}\beta }^{2}n_{j}(u^{\ ^{k-1}\alpha })$ \
following recurrent formulas (when a next ''shell'' extends in a compatible
form the previous ones; i.e. containing the previous constructions), being
computed as
\begin{eqnarray}
g_{\widehat{i}} &=&\epsilon _{\widehat{i}}e^{\psi (x^{\widehat{k}})},%
\mbox{\
for }\epsilon _{2}\psi ^{\bullet \bullet }+\epsilon _{3}\psi ^{\prime \prime
}=\Upsilon _{4};  \label{coeff} \\
h_{4} &=&\epsilon _{4}\ ^{0}h(x^{i})\ [\partial
_{v}f(x^{i},v)]^{2}|\varsigma (x^{i},v)|,h_{5}=\epsilon _{5}[f(x^{i},v)-\
^{0}f(x^{i})]^{2};  \notag
\end{eqnarray}%
\begin{eqnarray}
w_{i} &=&-\partial _{i}\varsigma (x^{i},v)/\partial _{v}\varsigma (x^{i},v),
\notag \\
n_{k} &=&\ ^{1}n_{k}(x^{i})+\ ^{2}n_{k}(x^{i})\int dv\ \varsigma (x^{i},v)
\notag \\
&&[\partial _{v}f(x^{i},v)]^{2}/[f(x^{i},v)-\ ^{0}f(x^{i})]^{3},  \notag \\
\mbox{\ for }\varsigma  &=&\ ^{0}\varsigma (x^{i})-\frac{\epsilon _{4}}{8}\
^{0}h(x^{i})\int dv\ \Upsilon _{2}(x^{k},v)\   \notag \\
&&\partial _{v}f(x^{i},v)[f(x^{i},v)-\ ^{0}f(x^{i})];  \notag
\end{eqnarray}%
\begin{eqnarray*}
h_{6} &=&\epsilon _{6}\ _{1}^{0}h(u^{\alpha })\ [\partial _{\ ^{1}v}\
^{1}f(u^{\ \alpha },\ ^{1}v)]^{2}|\ ^{1}\varsigma (u^{\ \alpha },\ ^{1}v)|,
\\
h_{7} &=&\epsilon _{7}[\ ^{1}f(u^{\ \alpha },\ ^{1}v)-\ _{1}^{0}f(u^{\
\alpha })]^{2}; \\
w_{\beta } &=&-\partial _{\beta }\ ^{1}\varsigma (u^{\ \alpha },\
^{1}v)/\partial _{\ ^{1}v}\ ^{1}\varsigma (u^{\ \alpha },\ ^{1}v), \\
n_{\beta } &=&\ ^{1}n_{\beta }(u^{\ \alpha })+\ ^{2}n_{\beta }(u^{\ \alpha
})\int d\ ^{1}v\ \ ^{1}\varsigma (u^{\ \alpha },\ ^{1}v) \\
&&[\partial _{\ ^{1}v}\ ^{1}f(u^{\ \alpha },\ ^{1}v)]^{2}/[\ ^{1}f(u^{\
\alpha },\ ^{1}v)-\ _{1}^{0}f(u^{\ \alpha })]^{3}, \\
\mbox{\ for }\ ^{1}\varsigma  &=&\ _{1}^{0}\varsigma (u^{\ \alpha })-\frac{%
\epsilon _{6}}{8}\ _{1}^{0}h(u^{\ \alpha })\int d\ ^{1}v\ \ ^{1}\Upsilon
_{2}(u^{\ \alpha },\ ^{1}v)\  \\
&&[\partial _{\ ^{1}v}\ ^{1}f(u^{\ \alpha },\ ^{1}v)][\ ^{1}f(u^{\ \alpha
},\ ^{1}v)-\ _{1}^{0}f(u^{\ \alpha })];
\end{eqnarray*}%
\begin{eqnarray*}
h_{8} &=&\epsilon _{8}\ _{2}^{0}h(u^{\ ^{1}\alpha })\ [\partial _{\ ^{2}v}\
^{2}f(u^{\ ^{1}\alpha },\ ^{2}v)]^{2}|\ ^{2}\varsigma (u^{\ ^{1}\alpha },\
^{2}v)|, \\
h_{9} &=&\epsilon _{9}[\ ^{2}f(u^{\ ^{1}\alpha },\ ^{2}v)-\ _{2}^{0}f(u^{\
^{1}\alpha })]^{2}; \\
w_{\ ^{1}\beta } &=&-\partial _{\ ^{1}\beta }\ ^{2}\varsigma (u^{\
^{1}\alpha },\ ^{2}v)/\partial _{\ ^{2}v}\ ^{2}\varsigma (u^{\ ^{1}\alpha
},\ ^{2}v), \\
n_{\ ^{1}\beta } &=&\ ^{1}n_{\ ^{1}\beta }(u^{\ ^{1}\alpha })+\ ^{2}n_{\
^{1}\beta }(u^{\ ^{1}\alpha })\int d\ ^{2}v\ \ ^{2}\varsigma (u^{\
^{1}\alpha },\ ^{2}v) \\
&&[\partial _{\ ^{2}v}\ ^{2}f(u^{\ ^{1}\alpha },\ ^{2}v)]^{2}/[\ ^{2}f(u^{\
^{1}\alpha },\ ^{2}v)-\ _{2}^{0}f(u^{\ ^{1}\alpha })]^{3}, \\
\mbox{\ for }\ ^{2}\varsigma  &=&\ _{2}^{0}\varsigma (u^{\ ^{1}\alpha })-%
\frac{\epsilon _{8}}{8}\ _{1}^{0}h(u^{\ ^{1}\alpha })\int d\ ^{2}v\ \
^{2}\Upsilon _{2}(u^{\ ^{1}\alpha },\ ^{2}v)\  \\
&&[\partial _{\ ^{2}v}\ ^{2}f(u^{\ ^{1}\alpha },\ ^{2}v)][\ ^{2}f(u^{\
^{1}\alpha },\ ^{2}v)-\ _{2}^{0}f(u^{\ ^{1}\alpha })];
\end{eqnarray*}%
\begin{equation*}
.........
\end{equation*}%
\begin{eqnarray*}
h_{4+2k} &=&\epsilon _{4+2k}\ _{k}^{0}h(u^{\ ^{k-1}\alpha })\ [\partial _{\
^{k}v}\ ^{k}f(u^{\ ^{k-1}\alpha },\ ^{k}v)]^{2}|\ ^{k}\varsigma (u^{\
^{k-1}\alpha },\ ^{k}v)|, \\
h_{5+2k} &=&\epsilon _{5+2k}[\ ^{k}f(u^{\ ^{k-1}\alpha },\ ^{k}v)-\
_{k}^{0}f(u^{\ ^{k-1}\alpha })]^{2};
\end{eqnarray*}%
\begin{eqnarray*}
w_{\ ^{k-1}\beta } &=&-\partial _{\ ^{k-1}\beta }\ ^{k}\varsigma (u^{\
^{k-1}\alpha },\ ^{k}v)/\partial _{\ ^{k}v}\ ^{k}\varsigma (u^{\
^{k-1}\alpha },\ ^{k}v), \\
n_{\ ^{k-1}\beta } &=&\ ^{1}n_{\ ^{k-1}\beta }(u^{\ ^{k-1}\alpha })+\
^{2}n_{\ ^{k-1}\beta }(u^{\ ^{k-1}\alpha })\int d\ ^{k}v\ \ ^{k}\varsigma
(u^{\ ^{k-1}\alpha },\ ^{k}v) \\
&&[\partial _{\ ^{k}v}\ ^{k}f(u^{\ ^{k-1}\alpha },\ ^{k}v)]^{2}/[\
^{k}f(u^{\ ^{k-1}\alpha },\ ^{k}v)-\ _{k}^{0}f(u^{\ ^{k-1}\alpha })]^{3}, \\
\mbox{\ for }\ ^{k}\varsigma  &=&\ _{k}^{0}\varsigma (u^{\ ^{k-1}\alpha })-%
\frac{\epsilon _{4+2k}}{8}\ _{k}^{0}h(u^{\ ^{k-1}\alpha })\int d\ ^{k}v\ \
^{k}\Upsilon _{2}(u^{\ ^{k-1}\alpha },\ ^{k}v)\  \\
&&[\partial _{\ ^{k}v}\ ^{k}f(u^{\ ^{k-1}\alpha },\ ^{k}v)][\ ^{k}f(u^{\
^{k-1}\alpha },\ ^{k}v)-\ _{k}^{0}f(u^{\ ^{k-1}\alpha })];
\end{eqnarray*}%
and
\begin{eqnarray}
\mathbf{e}_{k}\omega  &=&\partial _{k}\omega +w_{k}\partial _{v}\omega
+n_{k}\partial \omega /\partial y^{5}=0,  \label{confcondk} \\
\mathbf{e}_{\alpha }\ ^{1}\omega  &=&\partial _{\alpha }\ ^{1}\omega
+w_{\alpha }\partial \ ^{1}\omega /\partial \ ^{1}v+n_{\alpha }\partial \
^{1}\omega /\partial y^{5}=0,  \notag \\
\mathbf{e}_{\ ^{1}\alpha }\ ^{2}\omega  &=&\partial _{\ ^{1}\alpha }\
^{2}\omega +w_{\ ^{1}\alpha }\partial \ ^{2}\omega /\partial \ ^{2}v+n_{\
^{1}\alpha }\partial \ ^{2}\omega /\partial y^{7}=0,  \notag \\
&&......  \notag \\
\mathbf{e}_{\ ^{k-1}\alpha }\ ^{k}\omega  &=&\partial _{\ ^{k-1}\alpha }\
^{k}\omega +w_{\ ^{k-1}\alpha }\partial \ ^{k}\omega /\partial \ ^{k}v+n_{\
^{k-1}\alpha }\partial \ ^{k}\omega /\partial y^{5+2k}=0,  \notag
\end{eqnarray}%
when the solutions for the Levi--Civita connection are selected by
additional constraints%
\begin{eqnarray}
\partial _{v}w_{i}&=&\mathbf{e}_{i}\ln |h_{4}|,\mathbf{e}_{k}w_{i}=\mathbf{e}%
_{i}w_{k},\  \partial _{v}n_{i}=0,\ \partial _{i}n_{k}=\partial _{k}n_{i};
\label{lccond} \\
\partial _{\ ^{1}v}w_{\alpha }&=&\mathbf{e}_{\alpha }\ln |h_{6}|,\mathbf{e}%
_{\alpha }w_{\beta }=\mathbf{e}_{\beta }w_{\alpha }, \partial _{\
^{1}v}n_{\alpha }=0,\ \partial _{\alpha }n_{\beta }=\partial _{\beta
}n_{\alpha };  \notag \\
\partial _{\ ^{2}v}w_{\ ^{1}\alpha } &=&\mathbf{e}_{\ ^{1}\alpha }\ln
|h_{8}|,\mathbf{e}_{\ ^{1}\alpha }w_{\ ^{1}\beta }=\mathbf{e}_{\ ^{1}\beta
}w_{\ ^{1}\alpha },  \notag \\
&&\partial _{\ ^{2}v}n_{\ ^{1}\alpha }=0,\ \partial _{\ ^{1}\alpha }n_{\
^{1}\beta }=\partial _{\ ^{1}\beta }n_{\ ^{1}\alpha };  \notag \\
&&.....  \notag \\
\partial _{\ ^{k}v}w_{\ ^{k-1}\alpha } &=&\mathbf{e}_{\ ^{1}\alpha }\ln
|h_{4+2k}|,\mathbf{e}_{\ ^{k-1}\alpha }w_{\ ^{k-1}\beta }=\mathbf{e}_{\
^{k-1}\beta }w_{\ ^{k-1}\alpha },  \notag \\
&&\partial _{\ ^{k}v}n_{\ ^{k-1}\alpha }=0,\ \partial _{\ ^{k-1}\alpha }n_{\
^{k-1}\beta }=\partial _{\ ^{k-1}\beta }n_{\ ^{k-1}\alpha }.  \notag
\end{eqnarray}
\end{theorem}

Following above Theorem,\footnote{%
a similar Theorem was formulated in Ref. \cite{vgseg} for four and five
dimensions; some formulas and conditions have to be repeated in this work
because they are used for high dimension generalizations} we express the
solutions of Einstein equations in high dimensional gravity in a most
general form, presenting in formulas all classes of generating and
integration functions and stating all constraints selecting Einstein spaces.
We choose a ''two by two' increasing of spacetime dimensions in formulas
because this provides us a simplest way for generating ''non--Killing''
solutions characterized by certain types of two dimensional conformal
factors depending on two ''anisotropic'' coordinates and the rest ones being
considered as parameters. This allows us to associate to the constructed
classes of solutions certain hierarchies of two--dimensional conformal
symmetries with corresponding invariants and to derive associated solitonic
hierarchies and bi--Hamiltonian structures as we elaborated for four
dimensional spaces in Refs. \cite{vaam,avcf}. The length of this article
does not give us a possibility to present such nonlinear wave developments
for high dimensional gravity.

We have to state certain boundary/ symmetry/ topology conditions and define
in explicit form the integration functions and systems of first order
partial differential equations of type (\ref{lccond}). This is necessary
when we are interested to construct some explicit classes of exact solutions
of Einstein equations (\ref{einst1}) which are related to some physically
important four and higher dimensional metrics. For instance, such high
dimension solutions can be constructed to contain wormholes \cite%
{vsing2,vsing3} and/or to model (non) holonomic Ricci flows of various types
of gravitational solitonic pp--wave, ellipsoid etc configurations \cite%
{vrfs1,vrfs2,vrfs3,vrfs4,vrfs5}. Perhaps all classes of exact solutions
presented in the above mentioned references and (for instance, reviewed in)
Refs. \cite{kramer,bic,ijgmmp,vncg,vsgg,vrflg} can be found as certain
particular cases of metrics (\ref{ansgensol}) or certain equivalently
redefined ones. In this article, we shall emphasize the geometric background
of the anholonomic deformation method for construction high dimensional
exact solutions in gravity and refer readers to cited works, for details and
physical applications.

Any (pseudo) Riemannian metric $\mathbf{g}_{\ ^{k}\alpha ^{\prime }\
^{k}\beta ^{\prime }}(u^{\ ^{k}\gamma ^{\prime }})$ (\ref{metrml}) depending
in general on all $5+2k$ local coordinates on $\ \mathbf{V}$ can be
parametrized in a form $\mathbf{g}_{\ ^{k}\alpha \ ^{k}\beta }(u^{\
^{k}\gamma })$ (\ref{ansgensol}) using transforms of coefficients of metric
 $\mathbf{g}_{\ ^{k}\alpha \ ^{k}\beta }(u^{\ ^{k}\gamma })=e_{\ \ ^{k}\alpha
}^{\ ^{k}\alpha ^{\prime }}e_{\ \ ^{k}\beta }^{\ ^{k}\beta ^{\prime }}%
\mathbf{g}_{\ ^{k}\alpha ^{\prime }\ ^{k}\beta ^{\prime }}(u^{\ ^{k}\gamma
^{\prime }})$
 under vielbein transforms $e_{\ ^{k}\alpha }=e_{\ \ ^{k}\alpha }^{\
^{k}\alpha ^{\prime }}e_{\ ^{k}\alpha ^{\prime }}$ preserving a chosen shell
structure.\footnote{%
We have to solve certain systems of quadratic algebraic equations and define
some $e_{\ ^{\ k}\alpha }^{\ ^{k}\alpha ^{\prime }}(u^{\ ^{k}\beta }),$ for
given coefficients of any (\ref{ansgensol}) and (\ref{metrml}), choosing a
convenient system of coordinates $u^{\ ^{k}\alpha ^{\prime }}=u^{\
^{k}\alpha ^{\prime }}(u^{\ ^{k}\beta });$ to present the a ''shell
structure'' is convenient for purposes to simplify proofs of theorems; in
general, under arbitrary frame/coordinate transform, all ''shells'' mix each
with others.} If the coefficients of such metrics satisfy the conditions of
Main Theorem, they define general solutions of Einstein equations for any
type of sources which can be parametrized in a formally diagonalized (with
shell conditions) form (\ref{source}), with respect to certain classes of
nonholonomic frames of reference. By frame transforms such parametrizations
can be defined for various types of physically important energy--momentum
tensors, cosmological constants (in general, with anisotropic
polariziations), and for vacuum configurations.

For the case $k=0,$ the proof of Theorem \ref{mth} is outlined in Ref. \cite%
{vgseg} using the anholonomic deformation method which was originally
proposed in Refs. \cite{vpoland,vapny,vhep2}. There were published a series
of reviews and generalizations of the method, see \cite%
{ijgmmp,vncg,vsgg,vrflg}, when the solutions of gravitational field
equations in different types of commutative and noncommutative gravity and
Ricci flow theories contain at least one Killing vector symmetry. Such
higher dimension Einstein metrics with Killing symmetries, are generated if
in the conditions of Theorem \ref{mth} there are considered $\omega ,\
^{1}\omega ,\ldots ,\ ^{k}\omega =1.$

Summarizing the results provided in Sections \ref{s2}--\ref{s4} (they may be
considered  also as a review for the 'higher dimension' version of the
anholonomic deformation method of constructing exact solutions in gravity%
\footnote{%
this work is organized as following: In Section \ref{s2}, there are outlined
some necessary results from the geometry of nigher order nonholonomic
manifolds. Section \ref{s3} contains the system of partial differential
equations (PDE) to which the Einstein equations can be transformed under
nonholonomic frame transforms and deformations. In Section \ref{s4}, we
prove that it is possible to construct very general classes of exact
solutions with Killing symmetries for such PDE in high dimensional gravity.
We also show that it is possible to generate ''non--Killing'' solutions in
most general forms if we consider nonholonomically deformed conformal
symmetries. We conclude and discuss the results in Section \ref{s5}.}), we
get a proof of Main Theorem \ref{mth}.

Finally, we note that even we shall present a number of key results and some
technical details, we shall not repeat explicit computations for
coefficients of tensors and connections presented for the ''Killing case'' $%
k=0$ in our previous works \cite{vmth1,vmth2,vmth3}, see also reviews and
generalizations in \cite{ijgmmp,vncg}.

\section{Higher Order Nonholonomic Manifolds}

\label{s2} In this section, we outline the geometry of higher order
nonholonomic manifolds which, for simplicity, \ will be modelled as (pseudo)
Riemannian manifolds with higher order ''shell'' structure of dimensions $\
^{k}n=n+\ m+\ ^{1}m+...+\ ^{k}m.$ Certain geometric ideas and constructions
originate from the geometry of higher order Lagrange--Finsler and
Hamilton--Cartan spaces defined on higher order (co) tangent classical and
quantum bundles \cite{m1,m2,m3,av}. Such nonholonomic structures were
investigated for models of (super) strings in higher order anisotropic
(super) spaces \cite{vnp1} and for anholonomic higher order Clifford/ spinor
bundles \cite{vsp1,vv1,vst}. In explicit form, the (super) gravitational
gauge field equations and conservations laws were analyzed in Refs. \cite%
{vgauge,vvc,vd1,vd2}.

\subsection{Higher order N--adapted frames and metrics}

Our geometric spacetime arena is defined by (pseudo) Riemannian manifolds $\
^{k}\mathbf{V}$ enabled with nonholonomic distributions (which can be
prescribed in any convenient for our purposes form like we can fix any
system of reference/ coordinates).

\begin{definition}
A manifold $\ ^{k}\mathbf{V,}$ $\dim \ ^{k}\mathbf{V=}\ \ ^{k}n,$ is higher
order nonholonomic (equivalently, $k$--anholonomic) if its tangent bundle $%
T\ ^{k}\mathbf{V}$ is enabled with a Whitney sum of type
\begin{equation}
T\ \ ^{k}\mathbf{V=h}\ \ ^{k}\mathbf{V\oplus }\ v\ \ ^{k}\mathbf{V\mathbf{%
\oplus }\ \ }^{1}v\ \ ^{k}\mathbf{\mathbf{V\oplus }\ }\ ^{2}v\ \ ^{k}\mathbf{%
\mathbf{\mathbf{V}\oplus }\ \ldots \mathbf{\oplus }}\ ^{k}v\ \ ^{k}\mathbf{%
\mathbf{\mathbf{\mathbf{V}}}.}  \label{whitney}
\end{equation}%
For $k=0,$ we get a usual nonholonomic manifold (or, in this case,
N--anholonomic) enabled with nonlinear connection (N--connection) structure.
We say that a distribution (\ref{whitney}) defines a higher order
N--connection (equivalently, $\ ^{k}\mathbf{N}$--connection) structure.
\end{definition}

Locally, a $\ ^{k}\mathbf{N}$--connection is defined by its coefficients
$\ ^{k}\mathbf{N}=\{N_{i}^{a},N_{\alpha }^{\ ^{1}a},$ $N_{\ ^{1}\alpha }^{\
^{2}a}, \ldots , N_{\ ^{k-1}\alpha }^{\ ^{k}a}\},$ 
with 
$\{N_{i}^{a}\}\subset \{N_{\alpha }^{\ ^{1}a}\}\subset \{N_{\ ^{1}\alpha
}^{\ ^{2}a}\}\subset \ldots \subset \ N_{\ ^{k-1}\alpha }^{\ ^{k}a}\},$
when
\begin{eqnarray*}
\ ^{0}\mathbf{N} &=&N_{i}^{a}(u^{\alpha })dx^{i}\otimes \frac{\partial }{%
\partial y^{a}},\ ^{1}\mathbf{N}=N_{\beta }^{\ ^{1}a}(u^{\ ^{1}\alpha
})du^{\beta }\otimes \frac{\partial }{\partial y^{\ ^{1}a}}, \\
\ ^{2}\mathbf{N} &=&N_{\ ^{1}\beta }^{\ ^{2}a}(u^{\ ^{2}\alpha })du^{\
^{1}\beta }\otimes \frac{\partial }{\partial y^{\ ^{2}a}},\mathbf{\ldots ,}\
^{k}\mathbf{N}=N_{\ ^{k-1}\beta }^{\ ^{k}a}(u^{\ ^{k}\alpha })du^{\
^{k-1}\beta }\otimes \frac{\partial }{\partial y^{\ ^{k}a}}.
\end{eqnarray*}

It should be noted that for general coordinate transforms on $\ ^{k}\mathbf{V%
},$ there is a mixing of coefficients and coordinates. \footnote{%
We use boldface symbols for spaces (and geometric objects on such spaces)
enabled with a structure of N--coefficients.} For simplicity, we can work
with adapted coordinates when some sets of coordinates on a shell of lower
order are contained in a subset of coordinates on shells of higher order by
trivial extensions like $u^{\ ^{k-s}\alpha }\rightarrow u^{\ ^{k-s+1}\alpha
}=(u^{\ ^{k-s}\alpha },y^{\ ^{k-s+1}a}).$

\begin{proposition}
There is a class of N--adapted frames and dual (co-) frames (equivalently,
vielbeins) on $\ ^{k}\mathbf{V}$ which depend linearly on coefficients of $\
^{k}\mathbf{N}$--connection.
\end{proposition}

\begin{proof}
We construct such frames following recurrent formulas for $k=0,1,...,$ when
\begin{eqnarray}
\mathbf{e}_{\ ^{k}\nu } &=&\left( \mathbf{e}_{\ ^{k-1}v}=\frac{\partial }{%
\partial u^{\ ^{k-1}v}}-N_{\ ^{k-1}v}^{\ ^{k}a}\frac{\partial }{\partial
y^{\ ^{k}a}},e_{\ ^{k}a}=\partial _{^{k}a}=\frac{\partial }{\partial y^{\
^{k}a}}\right) ,  \label{dder} \\
\mathbf{e}^{\ ^{k}\mu } &=&\left( e^{^{k-1}\mu }=du^{^{k-1}\mu },\mathbf{e}%
^{\ ^{k}a}=dy^{\ ^{k}a}+N_{^{k-1}v}^{\ ^{k}a}du^{^{k-1}v}\right) .
\label{ddif}
\end{eqnarray}
$\ \square $ (End proof.)
\end{proof}

\vskip5pt

The vielbeins (\ref{dder}) satisfy the nonholonomy relations
\begin{equation}
\lbrack \mathbf{e}_{\ ^{k}\alpha },\mathbf{e}_{\ ^{k}\beta }]=\mathbf{e}_{\
^{k}\alpha }\mathbf{e}_{\ ^{k}\beta }-\mathbf{e}_{\ ^{k}\beta }\mathbf{e}_{\
^{k}\alpha }=w_{\ ^{k}\alpha \ ^{k}\beta }^{\ ^{k}\gamma }\mathbf{e}_{\
^{k}\gamma }  \label{anhrel}
\end{equation}%
with (antisymmetric) nontrivial anholonomy coefficients\newline
$w_{\ ^{k-1}\alpha \ ^{k}a}^{\ ^{k}b}=\partial N_{^{k-1}\alpha }^{\
^{k}b}/\partial u^{\ ^{k}a}$ and $w_{\ ^{k-1}\alpha \ ^{k-1}\beta }^{\
^{k}b}=\Omega _{\ ^{k-1}\alpha \ ^{k-1}\beta }^{\ ^{k}b},$ where
\begin{equation}
\Omega _{\ ^{k-1}\alpha \ ^{k-1}\beta }^{\ ^{k}b}=\mathbf{e}_{^{k-1}\beta
}\left( N_{^{k-1}\alpha }^{^{k}b}\right) -\mathbf{e}_{^{k-1}\alpha }\left(
N_{^{k-1}\beta }^{^{k}b}\right)  \label{ncurv}
\end{equation}%
are the coefficients of curvature $\ ^{k}\mathbf{\Omega }$ of \
N--connection $\ ^{k}\mathbf{N.}$ The particular holonomic/ integrable case
is selected by the integrability conditions $w_{\ ^{k-1}\alpha \ ^{k}a}^{\
^{k}b}=0.$

Any (pseudo) Riemannian metric $\mathbf{g}$ on $\ ^{k}\mathbf{V}$ can be
written in N--adapted form, we shall write in brief that $\mathbf{g}=\ ^{k}%
\mathbf{g=\{\mathbf{g}}_{\ ^{k}\beta \ ^{k}\gamma }\mathbf{\}},$ for \
\begin{eqnarray}
\ ^{k}\mathbf{g} &=&g_{\ ^{k-1}\beta \ ^{k-1}\gamma }(u^{\ ^{k-1}\alpha
})e^{^{k-1}\beta }\otimes e^{^{k-1}\gamma }+h_{\ ^{k}a\ ^{k}b}(u^{\
^{k}\alpha })\mathbf{e}^{\ ^{k}a}\otimes \mathbf{e}^{\ ^{k}b}  \label{gendm}
\\
&=&g_{ij}(x^{k})e^{i}\otimes e^{j}+h_{ab}(u^{\alpha })\mathbf{e}^{a}\otimes
\mathbf{e}^{b}+h_{\ ^{1}a\ ^{1}b}(u^{\ ^{1}\alpha })\mathbf{e}^{\
^{1}a}\otimes \mathbf{e}^{\ ^{1}b}  \notag \\
&&+h_{\ ^{2}a\ ^{2}b}(u^{\ ^{2}\alpha })\mathbf{e}^{\ ^{2}a}\otimes \mathbf{e%
}^{\ ^{2}b}+\ldots +h_{\ ^{k}a\ ^{k}b}(u^{\ ^{k}\alpha })\mathbf{e}^{\
^{k}a}\otimes \mathbf{e}^{\ ^{k}b},  \notag
\end{eqnarray}%
for some N--adapted coefficients $\mathbf{g}_{\ ^{k}\beta \ ^{k}\gamma }%
\mathbf{=}\left[ g_{ij},h_{ab},h_{\ ^{1}a\ ^{1}b},\ldots ,h_{\ ^{k}a\ ^{k}b}%
\right] $ and $N_{\ ^{k-1}\alpha }^{\ ^{k}a}.$ For constructing exact
solutions in high dimensional gravity, it is convenient to work with such
N--adapted formulas for tensors' and connections' coefficients.

For instance, we get from (\ref{gendm}) a metric with a parametrization of
type (\ref{ansgensol}) when all $\ ^{k}\omega =1$ if we choose
\begin{eqnarray}
g_{ij} &=&diag[\epsilon _{1},g_{\widehat{i}}(x^{\widehat{k}%
})],h_{ab}=diag[h_{a}(x^{i},v)],\  \notag \\
 &&N_{k}^{4}=w_{k}(x^{i},v),N_{k}^{5}=n_{k}(x^{i},v);  \label{data1}  \\
h_{\ ^{1}a\ ^{1}b} &=&diag[h_{\ ^{1}a}(u^{\alpha },\ ^{1}v)],  \notag \\
&&N_{\beta }^{6}=w_{\beta }(u^{\alpha },\ ^{1}v),N_{\beta }^{7}=n_{\beta
}(u^{\alpha },\ ^{1}v);  \notag \\
h_{\ ^{2}a\ ^{2}b} &=&diag[h_{\ ^{2}a}(u^{\ ^{1}\alpha },\ ^{2}v)],  \notag
\\
&&N_{\ ^{1}\beta }^{8}=w_{\ ^{1}\beta }(u^{\ ^{1}\alpha },\ ^{2}v),N_{\
^{1}\beta }^{9}=n_{\ ^{1}\beta }(u^{\ ^{1}\alpha },\ ^{2}v);  \notag \\
&&...........  \notag \\
h_{\ ^{k}a\ ^{k}b} &=&diag[h_{\ ^{k}a}(u^{\ ^{k-1}\alpha },\ ^{k}v)],  \notag
\\
&&N_{\ ^{k-1}\beta }^{4+2k}=w_{\ ^{k-1}\beta }(u^{\ ^{k-1}\alpha },\
^{k}v),N_{\ ^{k-1}\beta }^{5+2k}=n_{\ ^{k-1}\beta }(u^{\ ^{k-1}\alpha },\
^{k}v).  \notag
\end{eqnarray}%
Such a metric has symmetries of $k+1$ Killing vectors, $e_{5}=\partial
/\partial y^{5},e_{7}=\partial /\partial y^{7},...,e_{5+2k}=\partial
/\partial y^{5+2k},$ because its coefficients do not depend on $y^{5},y^{7},$
$...y^{5+2k}.$ Introducing nontrivial $\ ^{k}\omega ^{2}(u^{\ ^{k}\alpha })$
depending also on $y^{5+2k},$ as multiples before $h_{\ ^{k}a},$ we get
N--adapted parametrizations, up to certain frame/ coordinate transforms, for
all metrics on $\ ^{k}\mathbf{V}.$ In Section \ref{s3}, we shall define the
equations which must satisfy the coefficients of N--adapted metrics when (%
\ref{data1}) will generate exact solutions of Einstein equations.

\subsection{N--adapted deformations of the Levi--Civita connection}

By straightforward computations, it is a cumbersome task to prove using the
Levi--Civita connection\footnote{%
a unique one, which is metric compatible and with zero torsion, and
completely defined by the metric structure} $\ ^{k}\nabla $) that the
Einstein equations (\ref{einst1}) on higher dimensional spacetimes are
solved by metrics of type (\ref{ansgensol}). We are going to show explicitly
that general solutions for $\ ^{k}\nabla $ can be constructed passing three
steps:\ 1) to adapt our constructions to N--adapted frames of type $\mathbf{e%
}_{\alpha }$ (\ref{dder}) \ and $\mathbf{e}^{\mu }$ (\ref{ddif}); 2) to use
as an auxiliary tool (we emphasize, in Einstein gravity and its
generalizations on high dimensional (pseudo) Riemannian manifolds) a new
type of linear connection $\ ^{k}\widehat{\mathbf{D}}=\{\ \mathbf{\hat{\Gamma%
}}_{\ ^{k}\beta \ ^{k}\gamma }^{\ ^{k}\alpha }\},$ also uniquely defined by
the metric structure; 3) To constrain the integral varieties of general
solutions in such a form that $\ ^{k}\widehat{\mathbf{D}}\rightarrow \
^{k}\nabla .$

\begin{definition}
A distinguished connection $\ ^{k}\mathbf{D}$ (in brief, d--connection) on $%
\ ^{k}\mathbf{V}$ is a linear connection preserving under parallelism a
conventional horizontal and $k$--vertical splitting (in brief, h-- and
v--splitting) induced by \ $\ ^{k}\mathbf{N}$--connection structure (\ref%
{whitney}).
\end{definition}

We note that the Levi--Civita connection $\ ^{k}\nabla ,$ for which $\
^{k}\nabla \ ^{k}\mathbf{g=0}$ and $\ ^{k}\mathcal{T}^{\ ^{k}\alpha
}\doteqdot \ ^{k}\nabla \mathbf{e}^{\ ^{k}\alpha }=0,$ is not a
d--connection because, in general, it is not adapted to a N--splitting
defined by a Whitney sum (\ref{whitney}). So, in order to elaborate
self--consistent geometric/physical models adapted to a N--connection it is
necessary to work with d--connections.

\begin{theorem}
\label{th1}There is a unique canonical d--connection $\ ^{k}\widehat{\mathbf{%
D}}$ satisfying the condition $\ ^{k}\widehat{\mathbf{D}}\ ^{k}\mathbf{g=}0$
and with vanishing ''pure'' horizontal and vertical torsion coefficients, i.
e. $\widehat{T}_{\ jk}^{i}=0$ and $\widehat{T}_{\ \ ^{k}b\ ^{k}c}^{\
^{k}a}=0,$ see (below) formulas (\ref{dtor}).
\end{theorem}

\begin{proof}
Let us define $\ ^{k}\widehat{\mathbf{D}}$ as a 1--form
\begin{equation}
\widehat{\mathbf{\Gamma }}_{\ \ ^{k}\beta }^{\ ^{k}\alpha }=\widehat{\mathbf{%
\Gamma }}_{\ \ ^{k}\beta \ ^{k}\gamma }^{\ ^{k}\alpha }\mathbf{e}^{\
^{k}\gamma }  \label{cdcf}
\end{equation}%
with $\widehat{\mathbf{\Gamma }}_{\ \ ^{k}\alpha \ ^{k}\beta }^{\ ^{k}\gamma
}= ( \widehat{L}_{jk}^{i},\widehat{L}_{bk}^{a},\widehat{C}_{jc}^{i},\widehat{%
C}_{bc}^{a};\widehat{L}_{\ ^{1}b\alpha }^{\ ^{1}a},\widehat{C}_{\beta \
^{1}c}^{\alpha },\widehat{C}_{\ ^{1}b\ ^{1}c}^{\ ^{1}a};\ldots ;\widehat{L}%
_{\ ^{k}b\ ^{k-1}\alpha }^{\ ^{k}a},$\newline
$\widehat{C}_{\ ^{k-1}\beta \ ^{k}c}^{\ ^{k-1}\alpha },\widehat{C}_{\ ^{k}b\
^{k}c}^{\ ^{k}a}),$  where%
 $\widehat{L}_{\ ^{1}\beta \ ^{1}\gamma }^{\ ^{1}\alpha } =\widehat{\mathbf{%
\Gamma }}_{\beta \gamma }^{\alpha }=\left( \widehat{L}_{jk}^{i},\widehat{L}%
_{bk}^{a},\widehat{C}_{jc}^{i},\widehat{C}_{bc}^{a}\right);\
\widehat{L}_{\ ^{2}\beta \ ^{2}\gamma }^{\ ^{2}\alpha } =\widehat{\mathbf{%
\Gamma }}_{\ ^{1}\beta \ ^{1}\gamma }^{\ ^{1}\alpha };\ldots ;\widehat{L}_{\
^{k}\beta \ ^{k}\gamma }^{\ ^{k}\alpha }=\widehat{\mathbf{\Gamma }}_{\
^{k-1}\beta \ ^{k-1}\gamma }^{\ ^{k-1}\alpha },$
for%
\begin{eqnarray*}
\mathbf{g}_{\alpha \beta } &=&[g_{ij}(x^{k}),h_{ab}(u^{\gamma })],\mathbf{g}%
_{\ ^{1}\alpha \ ^{1}\beta }=[\mathbf{g}_{\alpha \beta }(u^{\gamma }),h_{\
^{1}a\ ^{1}b}(u^{\ ^{1}\gamma })], \\
\mathbf{g}_{\ ^{2}\alpha \ ^{2}\beta } &=&[\mathbf{g}_{\ ^{1}\alpha \
^{1}\beta }(u^{\ ^{1}\gamma }),h_{\ ^{2}a\ ^{2}b}(u^{\ ^{2}\gamma })],\ldots
, \\
\mathbf{g}_{\ ^{k}\alpha \ ^{k}\beta } &=&[\mathbf{g}_{\ ^{k-1}\alpha \
^{k-1}\beta }(u^{\ ^{k-1}\gamma }),h_{\ ^{k}a\ ^{k}b}(u^{\ ^{k}\gamma })],
\end{eqnarray*}%
\begin{eqnarray}
\mbox{where } &&\widehat{L}_{jk}^{i} =
\frac{1}{2}g^{ir}\left( \mathbf{e}_{k}g_{jr}+\mathbf{%
e}_{j}g_{kr}-\mathbf{e}_{r}g_{jk}\right) ,  \label{candcon} \\
&&\widehat{L}_{bk}^{a} = e_{b}(N_{k}^{a})+\frac{1}{2}h^{ac}\left( \mathbf{e}%
_{k}h_{bc}-h_{dc}\ e_{b}N_{k}^{d}-h_{db}\ e_{c}N_{k}^{d}\right) ,  \notag \\
&&\widehat{C}_{jc}^{i} = \frac{1}{2}g^{ik}e_{c}g_{jk},\
\widehat{C}_{bc}^{a} = \frac{1}{2}h^{ad}\left(
e_{c}h_{bd}+e_{c}h_{cd}-e_{d}h_{bc}\right) ,  \notag
\end{eqnarray}%
\begin{eqnarray*}
\widehat{L}_{\ ^{1}b\alpha }^{\ ^{1}a} &=&e_{\ ^{1}b}(N_{\ \alpha }^{\
^{1}a})+\frac{1}{2}h^{\ ^{1}a\ ^{1}c}(\mathbf{e}_{\alpha }h_{\ ^{1}b\
^{1}c}-h_{\ ^{1}d\ ^{1}c}\ e_{\ ^{1}b}N_{\ \alpha }^{\ ^{1}d} \\
&&-h_{\ ^{1}d\ ^{1}b}\ e_{\ ^{1}c}N_{\ \alpha }^{\ ^{1}d}),\
\widehat{C}_{\beta \ ^{1}c}^{\alpha } = \frac{1}{2}g^{\alpha \gamma }e_{\
^{1}c}g_{\beta \gamma },\  \\
\widehat{C}_{\ ^{1}b\ ^{1}c}^{\ ^{1}a} &=& \frac{1}{2}h^{\ ^{1}a\
^{1}d}\left( e_{\ ^{1}c}h_{\ ^{1}b\ ^{1}d}+e_{\ ^{1}c}h_{\ ^{1}c\
^{1}d}-e_{\ ^{1}d}h_{\ ^{1}b\ ^{1}c}\right) ,
\end{eqnarray*}%
\begin{eqnarray*}
\widehat{L}_{\ ^{2}b\ ^{1}\alpha }^{\ ^{2}a} &=&e_{\ ^{2}b}(N_{\ \
^{1}\alpha }^{\ ^{2}a})+\frac{1}{2}h^{\ ^{2}a\ ^{2}c}(\mathbf{e}_{\
^{1}\alpha }h_{\ ^{2}b\ ^{2}c}-h_{\ ^{2}d\ ^{2}c}\ e_{\ ^{2}b}N_{\ \
^{1}\alpha }^{\ ^{2}d} \\
&&-h_{\ ^{2}d\ ^{2}b}\ e_{\ ^{2}c}N_{\ \ ^{1}\alpha }^{\ ^{2}d}),\
\widehat{C}_{\ ^{1}\beta \ ^{2}c}^{\ ^{1}\alpha } =\frac{1}{2}g^{\
^{1}\alpha \ ^{1}\gamma }e_{\ ^{2}c}g_{\ ^{1}\beta \ ^{1}\gamma }, \\
\ \widehat{C}_{\ ^{2}b\ ^{2}c}^{\ ^{2}a} &=&\frac{1}{2}h^{\ ^{2}a\
^{2}d}\left( e_{\ ^{2}c}h_{\ ^{2}b\ ^{2}d}+e_{\ ^{2}c}h_{\ ^{2}c\
^{2}d}-e_{\ ^{2}d}h_{\ ^{2}b\ ^{2}c}\right) ,
\end{eqnarray*}%
\begin{equation*}
.......
\end{equation*}%
\begin{eqnarray*}
\widehat{L}_{\ ^{k}b\ ^{k-1}\alpha }^{\ ^{k}a} &=&e_{\ ^{k}b}(N_{\
^{k-1}\alpha }^{\ ^{k}a})+\frac{1}{2}h^{\ ^{k}a\ ^{k}c}(\mathbf{e}_{\
^{k-1}\alpha }h_{\ ^{k}b\ ^{k}c}-h_{\ ^{k}d\ ^{k}c}\ e_{\ ^{k}b}N_{\ \
^{k-1}\alpha }^{\ ^{k}d}- \\
&&h_{\ ^{k}d\ ^{k}b}\ e_{\ ^{k}c}N_{\ \ ^{k-1}\alpha }^{\ ^{k}d}),\
\widehat{C}_{\ ^{k-1}\beta \ ^{k}c}^{\ ^{k-1}\alpha } = \frac{1}{2}g^{\
^{k-1}\alpha\ ^{k-1}\gamma }e_{\ ^{k}c}g_{\ ^{k-1}\beta \ ^{k-1}\gamma},\
\\
\widehat{C}_{\ ^{k}b\ ^{k}c}^{\ ^{k}a} &=&\frac{1}{2}h^{\ ^{k}a\
^{k}d}\left( e_{\ ^{k}c}h_{\ ^{k}b\ ^{k}d}+e_{\ ^{k}c}h_{\ ^{k}c\
^{k}d}-e_{\ ^{k}d}h_{\ ^{k}b\ ^{k}c}\right) .
\end{eqnarray*}%
It follows by straightforward verifications that $\ ^{k}\widehat{\mathbf{D}}%
\ ^{k}\mathbf{g}=0,$ where this N--adapted metric compatibility condition
splits into
\begin{eqnarray}
\widehat{D}_{j}g_{kl} &=&0,\widehat{D}_{a}g_{kl}=0,\widehat{D}_{j}h_{ab}=0,%
\widehat{D}_{a}h_{bc}=0,  \label{metcomp} \\
\widehat{D}_{\gamma }g_{\alpha \beta } &=&0,\widehat{D}_{\ ^{1}a}g_{\alpha
\beta }=0,\widehat{D}_{\gamma }h_{\ ^{1}a\ ^{1}b}=0,\widehat{D}_{\
^{1}a}h_{\ ^{1}b\ ^{1}c}=0,  \notag \\
\widehat{D}_{\ ^{1}\gamma }g_{\ ^{1}\alpha \ ^{1}\beta } &=&0,\widehat{D}_{\
^{2}a}g_{\ ^{1}\alpha \ ^{1}\beta }=0,\widehat{D}_{\ ^{1}\gamma }h_{\ ^{2}a\
^{2}b}=0,\widehat{D}_{\ ^{2}a}h_{\ ^{2}b\ ^{2}c}=0,  \notag \\
&&......  \notag \\
\widehat{D}_{\ ^{k-1}\gamma }g_{\ ^{k-1}\alpha \ ^{k-1}\beta } &=&0,\widehat{%
D}_{\ ^{k}a}g_{\ ^{k-1}\alpha \ ^{k-1}\beta }=0,  \notag \\
\widehat{D}_{\ ^{k-1}\gamma }h_{\ ^{k}a\ ^{k}b} &=&0,\widehat{D}_{\
^{k}a}h_{\ ^{k}b\ ^{k}c}=0,  \notag
\end{eqnarray}%
where the covariant derivatives are computed using corresponding
coefficients, step by step, on every shell. The canonical d--connection
contains an induced torsion (completely determined by the coefficients of
metric, and respective N--connection coefficients)
\begin{equation}
\widehat{\mathcal{T}}^{\ ^{k}\alpha }=\widehat{\mathbf{T}}_{\ \ ^{k}\beta \
^{k}\gamma }^{\ ^{k}\alpha }\mathbf{e}^{\ ^{k}\beta }\wedge \mathbf{e}^{\
^{k}\gamma }\doteqdot \ ^{k}\widehat{\mathbf{D}}\mathbf{e}^{\ ^{k}\alpha }=d%
\mathbf{e}^{\ ^{k}\alpha }+\widehat{\mathbf{\Gamma }}_{\ \ ^{k}\beta }^{\
^{k}\alpha }\wedge \mathbf{e}^{\ ^{k}\beta },  \label{tors}
\end{equation}%
with coefficients
\begin{eqnarray}
\widehat{T}_{\ jk}^{i} &=&\widehat{L}_{\ jk}^{i}-\widehat{L}_{\ kj}^{i},\
\widehat{T}_{\ ja}^{i}=-\widehat{T}_{\ aj}^{i}=\widehat{C}_{\ ja}^{i},\ T_{\
ji}^{a}=-\Omega _{\ ji}^{a},\   \label{dtor} \\
\widehat{T}_{\ bi}^{a} &=&-\widehat{T}_{\ ib}^{a}=\frac{\partial N_{i}^{a}}{%
\partial y^{b}}-\widehat{L}_{\ bi}^{a},\ \widehat{T}_{\ bc}^{a}=\widehat{C}%
_{\ bc}^{a}-\widehat{C}_{\ cb}^{a};  \notag
\end{eqnarray}%
\begin{eqnarray*}
\widehat{T}_{\ \beta \gamma }^{\alpha } &=&\widehat{L}_{\ \ \beta \gamma
}^{\alpha }-\widehat{L}_{\ \ \gamma \beta }^{\alpha },\ \widehat{T}_{\ \beta
\ ^{1}a}^{\alpha }=-\widehat{T}_{\ ^{1}a\beta }^{\alpha }=\widehat{C}_{\
\beta \ ^{1}a}^{\alpha },\ T_{\ \beta \alpha }^{\ ^{1}a}=-\Omega _{\ \beta
\alpha }^{\ ^{1}a},\  \\
\widehat{T}_{\ \ ^{1}b\alpha }^{\ ^{1}a} &=&-\widehat{T}_{\ \alpha \
^{1}b}^{\ ^{1}a}=\frac{\partial N_{\alpha }^{\ ^{1}a}}{\partial y^{\ ^{1}b}}-%
\widehat{L}_{\ \ ^{1}b\alpha }^{\ ^{1}a},\ \widehat{T}_{\ \ ^{1}b\ ^{1}c}^{\
^{1}a}=\widehat{C}_{\ \ ^{1}b\ ^{1}c}^{\ ^{1}a}-\widehat{C}_{\ \ ^{1}c\
^{1}b}^{\ ^{1}a};
\end{eqnarray*}%
\begin{eqnarray*}
\widehat{T}_{\ \ ^{1}\beta \ ^{1}\gamma }^{\ ^{1}\alpha } &=&\widehat{L}_{\
\ \ ^{1}\beta \ ^{1}\gamma }^{\ ^{1}\alpha }-\widehat{L}_{\ \ \ ^{1}\gamma \
^{1}\beta }^{\ ^{1}\alpha },\ \widehat{T}_{\ \ ^{1}\beta \ ^{2}a}^{\
^{1}\alpha }=-\widehat{T}_{\ ^{2}a\ ^{1}\beta }^{\ ^{1}\alpha }=\widehat{C}%
_{\ \ ^{1}\beta \ ^{2}a}^{\ ^{1}\alpha },\  \\
T_{\ \ ^{1}\beta \ ^{1}\alpha }^{\ ^{2}a} &=&-\Omega _{\ \ ^{1}\beta \
^{1}\alpha }^{\ ^{2}a},\ \widehat{T}_{\ \ ^{2}b\ ^{1}\alpha }^{\ ^{2}a}=-%
\widehat{T}_{\ \ ^{1}\alpha \ ^{2}b}^{\ ^{2}a}=\frac{\partial N_{\
^{1}\alpha }^{\ ^{2}a}}{\partial y^{\ ^{2}b}}-\widehat{L}_{\ \ ^{2}b\
^{1}\alpha }^{\ ^{2}a}, \\
\ \widehat{T}_{\ \ ^{2}b\ ^{2}c}^{\ \ ^{2}a} &=&\widehat{C}_{\ \ ^{2}b\
^{2}c}^{\ ^{2}a}-\widehat{C}_{\ \ ^{2}c\ ^{2}b}^{\ ^{2}a};
\end{eqnarray*}%
\begin{equation*}
.....
\end{equation*}%
\begin{eqnarray*}
\widehat{T}_{\ \ ^{k-1}\beta \ ^{k-1}\gamma }^{\ ^{k-1}\alpha } &=&\widehat{L%
}_{\ \ \ ^{k-1}\beta \ ^{k-1}\gamma }^{\ ^{k-1}\alpha }-\widehat{L}_{\ \ \
^{k-1}\gamma \ ^{k-1}\beta }^{\ ^{k-1}\alpha }, \\
\ \widehat{T}_{\ \ ^{k-1}\beta \ ^{k}a}^{\ ^{k-1}\alpha } &=&-\widehat{T}_{\
^{k}a\ ^{k-1}\beta }^{\ ^{k-1}\alpha }=\widehat{C}_{\ \ ^{k-1}\beta \
^{k}a}^{\ ^{k-1}\alpha },\ T_{\ \ ^{k-1}\beta \ ^{k-1}\alpha }^{\
^{k}a}=-\Omega _{\ \ ^{k-1}\beta \ ^{k-1}\alpha }^{\ ^{k}a},\  \\
\widehat{T}_{\ \ ^{k}b\ ^{k-1}\alpha }^{\ ^{k}a} &=&-\widehat{T}_{\ \
^{k-1}\alpha \ ^{k}b}^{\ ^{k}a}=\frac{\partial N_{\ ^{k-1}\alpha }^{\ ^{k}a}%
}{\partial y^{\ ^{k}b}}-\widehat{L}_{\ \ ^{k}b\ ^{k-1}\alpha }^{\ ^{k}a},\
\\
\widehat{T}_{\ \ ^{k}b\ ^{k}c}^{\ \ \ ^{k}a} &=&\widehat{C}_{\ \ ^{k}b\
^{k}c}^{\ ^{k}a}-\widehat{C}_{\ \ ^{k}c\ ^{k}b}^{\ ^{k}a}.
\end{eqnarray*}%
Introducing values (\ref{candcon}) into (\ref{dtor}) we get that $\widehat{T}%
_{\ jk}^{i}=0$ and $\widehat{C}_{\ \ ^{k}b\ ^{k}c}^{\ ^{k}a}=0$ which
satisfy the conditions of this theorem. In general, other N--adapted torsion
coefficients (for instance, $\widehat{T}_{\ ja}^{i},\widehat{T}_{\ ji}^{a}$
and $\widehat{T}_{\ bi}^{a})$ are not zero. $\ \square $
\end{proof}

\vskip5pt

The torsion (\ref{tors}) is very different from that, for instance, in
Einstein--Cartan, string, or gauge gravity because we do not consider
additional field equations (algebraic or dynamical ones), see discussions in %
\cite{vncg,vsgg}. In our case, the nontrivial torsion coefficients are
related to anholonomy coefficients $w_{\ ^{k}\alpha \ ^{k}\beta }^{\
^{k}\gamma }$ in (\ref{anhrel}).

We can distinguish the covariant derivative $\ ^{k}\widehat{\mathbf{D}}$
determined by formulas (\ref{cdcf}) and (\ref{candcon}), in N--adapted to (%
\ref{whitney}) form, as $\widehat{\mathbf{D}}_{\ ^{k}\alpha }=\left(
\widehat{D}_{i},\widehat{D}_{a},\widehat{D}_{\ ^{1}a},...,\widehat{D}_{\
^{k}a}\right),$ where $\widehat{D}_{\ ^{k}a}$ are shell operators.

From Theorem \ref{th1}, we get:

\begin{corollary}
Any geometric construction for the canonical d--connection $\ \ ^{k}\widehat{%
\mathbf{D}}=\{\widehat{\mathbf{\Gamma }}_{\ \ ^{k}\alpha \ ^{k}\beta }^{\
^{k}\gamma }\}$ can be re--defined equivalently into a similar one with the
Levi--Civita connection $\ \ ^{k}\nabla =\{\Gamma _{\ \ ^{k}\alpha \
^{k}\beta }^{\ ^{k}\gamma }\}$ following formulas
\begin{equation}
\ \Gamma _{\ \ ^{k}\alpha \ ^{k}\beta }^{\ ^{k}\gamma }=\widehat{\mathbf{%
\Gamma }}_{\ \ ^{k}\alpha \ ^{k}\beta }^{\ ^{k}\gamma }+\ Z_{\ \ ^{k}\alpha
\ ^{k}\beta }^{\ ^{k}\gamma },  \label{deflc}
\end{equation}%
where the N--adapted coefficients $\ $of linear connections, $\Gamma _{\ \
^{k}\alpha \ ^{k}\beta }^{\ ^{k}\gamma },\ \widehat{\mathbf{\Gamma }}_{\ \
^{k}\alpha \ ^{k}\beta }^{\ ^{k}\gamma },$ and the distortion tensor $\ \
Z_{\ \ ^{k}\alpha \ ^{k}\beta }^{\ ^{k}\gamma }$ are determined in unique
forms by the coefficients of a metric $\mathbf{g}_{\ ^{k}\alpha \ ^{k}\beta
}.$
\end{corollary}

\begin{proof}
It is similar to that presented for vector bundles in Refs. \cite{ma1987,ma}
but in our case adapted for (pseudo) Riemannian nonholonomic manifolds, see
details in \cite{ijgmmp,vsgg,vrflg} and, in higher order form, in \cite%
{vnp1,vsp1,vst}. Here we present the N--adapted components of the distortion
tensor $Z_{\ \ ^{k}\alpha \ ^{k}\beta }^{\ ^{k}\gamma }$ computed as
\begin{eqnarray}
\ Z_{jk}^{a} &=&-\widehat{C}_{jb}^{i}g_{ik}h^{ab}-\frac{1}{2}\Omega
_{jk}^{a},~Z_{bk}^{i}=\frac{1}{2}\Omega _{jk}^{c}h_{cb}g^{ji}-\Xi _{jk}^{ih}~%
\widehat{C}_{hb}^{j},  \notag \\
Z_{bk}^{a} &=&~^{+}\Xi _{cd}^{ab}~\widehat{T}_{kb}^{c},\ Z_{kb}^{i}=\frac{1}{%
2}\Omega _{jk}^{a}h_{cb}g^{ji}+\Xi _{jk}^{ih}~\widehat{C}_{hb}^{j},\
Z_{jk}^{i}=0,  \label{deft} \\
\ Z_{jb}^{a} &=&-~^{-}\Xi _{cb}^{ad}~\widehat{T}_{jd}^{c},\ Z_{bc}^{a}=0,\
Z_{ab}^{i}=-\frac{g^{ij}}{2}\left[ \widehat{T}_{ja}^{c}h_{cb}+\widehat{T}%
_{jb}^{c}h_{ca}\right] ,  \notag \\
\ Z_{\ \beta \gamma }^{\ ^{1}a} &=&-\widehat{C}_{\beta \ ^{1}b}^{\alpha
}g_{\alpha \gamma }h^{\ ^{1}a\ ^{1}b}-\frac{1}{2}\Omega _{\beta \gamma }^{\
^{1}a},\
Z_{\ ^{1}b\gamma }^{\alpha } = \frac{1}{2}\Omega _{\beta \gamma }^{\
^{1}c}h_{\ ^{1}c\ ^{1}b}g^{\beta \alpha }-\Xi _{\beta \gamma }^{\alpha \tau
}~\widehat{C}_{\tau \ ^{1}b}^{\beta }, \notag \\
Z_{\ ^{1}b\gamma }^{\ ^{1}a} &=&~^{+}\Xi _{\ ^{1}c\ ^{1}d}^{\ ^{1}a\ ^{1}b}~%
\widehat{T}_{\gamma \ ^{1}b}^{\ ^{1}c},\ Z_{\beta \ ^{1}b}^{\alpha }=\frac{1%
}{2}\Omega _{\beta \gamma }^{\ ^{1}a}h_{\ ^{1}c\ ^{1}b}g^{\beta \alpha }+\Xi
_{\beta \gamma }^{\alpha \tau }~\widehat{C}_{\tau \ ^{1}b}^{\beta }, \notag \\
\ Z_{\beta \gamma }^{\alpha } &=&0,\ Z_{\beta \ ^{1}b}^{\ ^{1}a}=-~^{-}\Xi
_{\ ^{1}c\ ^{1}b}^{\ ^{1}a\ ^{1}d}~\widehat{T}_{\beta \ ^{1}d}^{\ ^{1}c},\
Z_{\ ^{1}b\ ^{1}c}^{\ ^{1}a}=0,\ \notag \\
Z_{\ ^{1}a\ ^{1}b}^{\alpha } &=&-\frac{g^{\alpha \beta }}{2}\left[ \widehat{T%
}_{\beta \ ^{1}a}^{\ ^{1}c}h_{\ ^{1}c\ ^{1}b}+\widehat{T}_{\beta \ ^{1}b}^{\
^{1}c}h_{\ ^{1}c\ ^{1}a}\right], \notag
\end{eqnarray}%
\begin{eqnarray*}
&&\ Z_{\ \ ^{1}\beta \ ^{1}\gamma }^{\ ^{2}a} =-\widehat{C}_{\ ^{1}\beta \
^{2}b}^{\ ^{1}\alpha }g_{\ ^{1}\alpha \ ^{1}\gamma }h^{\ ^{2}a\ ^{2}b}-\frac{%
1}{2}\Omega _{\ ^{1}\beta \ ^{1}\gamma }^{\ ^{2}a},\
Z_{\ ^{2}b\ ^{1}\gamma }^{\ ^{1}\alpha } = \widehat{L}_{4j}^{5}=\frac{1}{2}\partial _{v}n_{j},\\ && \frac{1}{2}\Omega _{\ ^{1}\beta
\ ^{1}\gamma }^{\ ^{2}c}h_{\ ^{2}c\ ^{2}b}g^{\ ^{1}\beta \ ^{1}\alpha }-\Xi
_{\ ^{1}\beta \ ^{1}\gamma }^{\ ^{1}\alpha \ ^{1}\tau }~\widehat{C}_{\
^{1}\tau \ ^{2}b}^{\ ^{1}\beta },\
Z_{\ ^{2}b\ ^{1}\gamma }^{\ ^{2}a} =\ ^{+}\Xi _{\ ^{2}c\ ^{2}d}^{\ ^{2}a\
^{2}b}~\widehat{T}_{\ ^{1}\gamma \ ^{2}b}^{\ ^{2}c},\  \\
&&Z_{\ ^{1}\beta \ ^{2}b}^{\ ^{1}\alpha } =\frac{1}{2}\Omega _{\ ^{1}\beta \
^{1}\gamma }^{\ ^{2}a}h_{\ ^{2}c\ ^{2}b}g^{\ ^{1}\beta \ ^{1}\alpha }+\Xi
_{\ ^{1}\beta \ ^{1}\gamma }^{\ ^{1}\alpha \ ^{1}\tau }~\widehat{C}_{\
^{1}\tau \ ^{2}b}^{\ ^{1}\beta }, \\
&&\ Z_{\ ^{1}\beta \ ^{1}\gamma }^{\ ^{1}\alpha } =0,\ Z_{\ ^{1}\beta \
^{2}b}^{\ ^{2}a}=-~^{-}\Xi _{\ ^{2}c\ ^{2}b}^{\ ^{2}a\ ^{2}d}~\widehat{T}_{\
^{1}\beta \ ^{2}d}^{\ ^{2}c},\ Z_{\ ^{2}b\ ^{2}c}^{\ ^{2}a}=0,\  \\
&&Z_{\ ^{2}a\ ^{2}b}^{\ ^{1}\alpha } =-\frac{g^{\ ^{1}\alpha \ ^{1}\beta }}{2%
}\left[ \widehat{T}_{\ ^{1}\beta \ ^{2}a}^{\ ^{2}c}h_{\ ^{2}c\ ^{2}b}+%
\widehat{T}_{\ ^{1}\beta \ ^{2}b}^{\ ^{2}c}h_{\ ^{2}c\ ^{2}a}\right] ,\\
&&..... \\
&&\ Z_{\ \ ^{k-1}\beta \ ^{k-1}\gamma }^{\ ^{k}a} =-\widehat{C}_{\
^{k-1}\beta \ ^{k}b}^{\ ^{k-1}\alpha }g_{\ ^{k-1}\alpha \ ^{k-1}\gamma }h^{\
^{k}a\ ^{k}b}-\frac{1}{2}\Omega _{\ ^{k-1}\beta \ ^{k-1}\gamma }^{\ ^{k}a},~
\\
&&Z_{\ ^{k}b\ ^{k-1}\gamma }^{\ ^{k-1}\alpha } =\frac{1}{2}\Omega _{\
^{k-1}\beta \ ^{k-1}\gamma }^{\ ^{k}c}h_{\ ^{k}c\ ^{k}b}g^{\ ^{k-1}\beta \
^{k-1}\alpha }-\Xi _{\ ^{k-1}\beta \ ^{k-1}\gamma }^{\ ^{k-1}\alpha \
^{k-1}\tau }~\widehat{C}_{\ ^{k-1}\tau \ ^{k}b}^{\ ^{k-1}\beta }, \\
&&Z_{\ ^{k}b\ ^{k-1}\gamma }^{\ ^{k}a} =\ ^{+}\Xi _{\ ^{k}c\ ^{k}d}^{\ ^{k}a\
^{k}b}~\widehat{T}_{\ ^{k-1}\gamma \ ^{k}b}^{\ ^{k}c},\  \\
&&Z_{\ ^{k-1}\beta \ ^{k}b}^{\ ^{k-1}\alpha } =\frac{1}{2}\Omega _{\
^{k-1}\beta \ ^{k-1}\gamma }^{\ ^{k}a}h_{\ ^{k}c\ ^{k}b}g^{\ ^{k-1}\beta \
^{k-1}\alpha }+\Xi _{\ ^{k-1}\beta \ ^{k-1}\gamma }^{\ ^{k-1}\alpha \
^{k-1}\tau }~\widehat{C}_{\ ^{k-1}\tau \ ^{k}b}^{\ ^{k-1}\beta },
\end{eqnarray*}%
\begin{eqnarray*}
\ Z_{\ ^{k-1}\beta \ ^{k-1}\gamma }^{\ ^{k-1}\alpha } &=&0,\ Z_{\
^{k-1}\beta \ ^{k}b}^{\ ^{k}a}=-~^{-}\Xi _{\ ^{k}c\ ^{k}b}^{\ ^{k}a\ ^{k}d}~%
\widehat{T}_{\ ^{k-1}\beta \ ^{k}d}^{\ ^{k}c},\ Z_{\ ^{k}b\ ^{k}c}^{\
^{k}a}=0,\\
Z_{\ ^{k}a\ ^{k}b}^{\ ^{k-1}\alpha } &=&-\frac{g^{\ ^{k-1}\alpha \
^{k-1}\beta }}{2}\left[ \widehat{T}_{\ ^{k-1}\beta \ ^{k}a}^{\ ^{k}c}h_{\
^{k}c\ ^{k}b}+\widehat{T}_{\ ^{k-1}\beta \ ^{k}b}^{\ ^{k}c}h_{\ ^{k}c\ ^{k}a}%
\right] ,
\end{eqnarray*}%
\begin{eqnarray*}
\mbox{ for } \Xi _{jk}^{ih} &=&\frac{1}{2}(\delta _{j}^{i}\delta
_{k}^{h}-g_{jk}g^{ih}),~^{\pm }\Xi _{cd}^{ab}=\frac{1}{2}(\delta
_{c}^{a}\delta _{d}^{b}\pm h_{cd}h^{ab}), \\
~^{\pm }\Xi _{\ ^{1}c\ ^{1}d}^{\ ^{1}a\ ^{1}b} &=&\frac{1}{2}(\delta _{\
^{1}c}^{\ ^{1}a}\delta _{\ ^{1}d}^{\ ^{1}b}\pm h_{\ ^{1}c\ ^{1}d}h^{\ ^{1}a\
^{1}b}),\ldots , \\
~^{\pm }\Xi _{\ ^{k}c\ ^{k}d}^{\ ^{k}a\ ^{k}b} &=&\frac{1}{2}(\delta _{\
^{k}c}^{\ ^{k}a}\delta _{\ ^{k}d}^{\ ^{k}b}\pm h_{\ ^{k}c\ ^{k}d}h^{\ ^{k}a\
^{k}b}),
\end{eqnarray*}%
where the necessary torsion coefficients are computed as in (\ref{dtor}).$\
\square $
\end{proof}

\vskip5pt

\begin{remark}
Hereafter, we shall omit certain details on shell components of formulas and
computations if that will not result in ambiguities. Such constructions are
similar to those presented in above Theorems and in Refs. \cite%
{ijgmmp,vncg,vsgg,vrflg,vmth1,vmth2,vmth3,vnp1,vsp1,vv1,vst,vvc,vd1,vd2}.
Some additional necessary formulas are given in Appendices.
\end{remark}

In four dimensions, the Einstein gravity can be equivalently formulated in
the so--called almost K\"{a}hler and Lagrange--Finsler variables, as we
considered in Refs. \cite{vrflg,vbrane,vkvnsst,vdqak}. Similarly, for higher
dimensions, we can use the canonical d--connection $\ ^{k}\widehat{\mathbf{D}%
}$ and its nonholonomic deformations for equivalent reformulations of extra
dimension gravity theories and as tools for generating constructing exact
solutions. In particular, imposing necessary type constraints, it is
possible to generate exact solutions of the Einstein equations for the
Levi--Civita connection $\ ^{k}\nabla .$

\section{N--adapted Einstein Equations}

\label{s3} In this section, we define the Riemannian, Ricci and Einstein
tensors for the canonical d--connection $\ ^{k}\widehat{\mathbf{D}}$ and
metric $\ ^{k}\mathbf{g}$ (\ref{gendm}) and derive the corresponding
gravitational field equations. We also formulate the general conditions when
the Einstein tensor for $\ ^{k}\widehat{\mathbf{D}}$ is equal to that for $\
^{k}\nabla .$

\subsection{Curvature of the canonical d--connection}

As for any linear connection, we can introduce:

\begin{definition}
The curvature of $\widehat{\mathbf{D}}$ is a 2--form 
$\widehat{\mathcal{R}}\doteqdot \widehat{\mathbf{D}}\widehat{\mathbf{\Gamma }%
}=d\widehat{\mathbf{\Gamma }}-\widehat{\mathbf{\Gamma }}\wedge \widehat{%
\mathbf{\Gamma }}.$ 
\end{definition}

In explicit form, the N--adapted coefficients can be computed using the
1--form (\ref{cdcf}),
\begin{equation}
\widehat{\mathcal{R}}_{~\quad ^{k}\beta }^{\ ^{k}\alpha }\doteqdot \widehat{%
\mathbf{D}}\widehat{\mathbf{\Gamma }}_{\ \ ^{k}\beta }^{^{k}\alpha }=d%
\widehat{\mathbf{\Gamma }}_{\ \ ^{k}\beta }^{^{k}\alpha }-\widehat{\mathbf{%
\Gamma }}_{\ \ ^{k}\beta }^{^{k}\gamma }\wedge \widehat{\mathbf{\Gamma }}_{\
\ ^{k}\gamma }^{^{k}\alpha }=\widehat{\mathbf{R}}_{\ \ ^{k}\beta \
^{k}\gamma \ ^{k}\tau }^{^{k}\alpha }\mathbf{e}^{\ ^{k}\gamma }\wedge
\mathbf{e}^{^{k}\tau },  \label{curv}
\end{equation}%
The N--adapted coefficients of curvature are parametrized in the form:%
\begin{eqnarray}
\widehat{\mathbf{R}}_{\ \ ^{k}\beta \ ^{k}\gamma \ ^{k}\tau }^{^{k}\alpha }
&=&\langle \widehat{R}_{\ \beta \gamma \tau }^{\alpha }=\{\widehat{R}_{\
hjk}^{i},\widehat{R}_{\ bjk}^{a},\widehat{R}_{\ jka}^{i},\widehat{R}_{\
bka}^{c},\widehat{R}_{\ jbc}^{i},\widehat{R}_{\ bcd}^{a}\};  \label{curvnc}
\\
\widehat{R}_{\ \ ^{1}\beta \ ^{1}\gamma \ ^{1}\tau }^{\ ^{1}\alpha } &=&\{%
\widehat{R}_{\ \beta \gamma \tau }^{\alpha },\widehat{R}_{\ ^{1}b\gamma \tau
}^{^{1}a},\widehat{R}_{\ \beta \gamma \ ^{1}a}^{\alpha },\widehat{R}_{\ \
^{1}b\gamma \ ^{1}a}^{\ ^{1}c},\widehat{R}_{\ \beta \ ^{1}b\ ^{1}c}^{\alpha
},\widehat{R}_{\ \ ^{1}b\ ^{1}c\ ^{1}d}^{\ ^{1}a};\};  \notag \\
\widehat{R}_{\ \ ^{2}\beta \ ^{2}\gamma \ ^{2}\tau }^{\ ^{2}\alpha } &=&\{%
\widehat{R}_{\ \ ^{1}\beta \ ^{1}\gamma \ ^{1}\tau }^{\ ^{1}\alpha },%
\widehat{R}_{\ ^{2}b\ ^{1}\gamma \ ^{1}\tau }^{^{2}a},\widehat{R}_{\ \
^{1}\beta \ ^{1}\gamma \ ^{2}a}^{\ ^{1}\alpha },\widehat{R}_{\ \ ^{2}b\
^{1}\gamma \ ^{2}a}^{\ ^{2}c},  \notag \\
&&\widehat{R}_{\ \ ^{1}\beta \ ^{2}b\ ^{2}c}^{\ ^{1}\alpha },\widehat{R}_{\
\ ^{2}b\ ^{2}c\ ^{2}d}^{\ ^{2}a}\};  \notag
\end{eqnarray}
\begin{eqnarray*}
&&.....  \notag \\
\widehat{R}_{\ \ ^{k}\beta \ ^{k}\gamma \ ^{k}\tau }^{\ ^{k}\alpha } &=&\{%
\widehat{R}_{\ \ ^{k-1}\beta \ ^{k-1}\gamma \ ^{k-1}\tau }^{\ ^{k-1}\alpha },%
\widehat{R}_{\ ^{k}b\ ^{k-1}\gamma \ ^{k-1}\tau }^{^{k}a},\widehat{R}_{\ \
^{k-1}\beta \ ^{k-1}\gamma \ ^{k}a}^{\ ^{k-1}\alpha },  \notag \\
&&\widehat{R}_{\ \ ^{k}b\ ^{k-1}\gamma \ ^{k}a}^{\ ^{k}c},\widehat{R}_{\ \
^{k-1}\beta \ ^{k}b\ ^{k}c}^{\ ^{k-1}\alpha },\widehat{R}_{\ \ ^{k}b\ ^{k}c\
^{k}d}^{\ ^{k}a}\}\rangle,  \notag
\end{eqnarray*}
where the values of such coefficients are provided in Appendix \ref{asa},
see Theorem \ref{thasa1} and formulas (\ref{dcurv}).

\begin{definition}
The Ricci tensor $Ric(\widehat{\mathbf{D}})=\{\widehat{\mathbf{R}}_{\alpha
\beta }\}$ of a canonical d--connec\-ti\-on $\widehat{\mathbf{D}}$ is
defined by contracting respectively the N--adapted coefficients of $\widehat{%
\mathbf{R}}_{\ \beta \gamma \delta }^{\alpha }$ (\ref{curv}), when $\widehat{%
\mathbf{R}}_{\alpha \beta }\doteqdot \widehat{\mathbf{R}}_{\ \alpha \beta
\tau }^{\tau }.$
\end{definition}

We formulate:

\begin{corollary}
The Ricci tensor of $\widehat{\mathbf{D}}$ is characterized by N--adapted
coefficients%
\begin{eqnarray}
\widehat{\mathbf{R}}_{\ ^{k}\alpha \ ^{k}\beta } &=&\{\widehat{R}_{ij},%
\widehat{R}_{ia},\ \widehat{R}_{ai},\ \widehat{R}_{ab};\widehat{R}_{\alpha
\beta },\widehat{R}_{\alpha \ ^{1}a},\ \widehat{R}_{\ ^{1}a\beta },\
\widehat{R}_{\ ^{1}a\ ^{1}b};  \label{driccia} \\
&&\widehat{R}_{\ ^{1}\alpha \ ^{1}\beta },\widehat{R}_{\ ^{1}\alpha \
^{2}a},\ \widehat{R}_{\ ^{2}a\ ^{1}\beta },\ \widehat{R}_{\ ^{2}a\
^{2}b};...;  \notag \\
&&\widehat{R}_{\ ^{k-1}\alpha \ ^{k-1}\beta },\widehat{R}_{\ ^{k-1}\alpha \
^{k}a},\ \widehat{R}_{\ ^{k}a\ ^{k-1}\beta },\ \widehat{R}_{\ ^{k}a\
^{k}b}\},  \notag
\end{eqnarray}%
\begin{eqnarray}
\mbox{where }
 \widehat{R}_{ij} &\doteqdot &\widehat{R}_{\ ijk}^{k},\ \ \widehat{R}%
_{ia}\doteqdot -\widehat{R}_{\ ika}^{k},\ \widehat{R}_{ai}\doteqdot \widehat{%
R}_{\ aib}^{b},\ \widehat{R}_{ab}\doteqdot \widehat{R}_{\ abc}^{c};
\label{dricci} \\
\widehat{R}_{\alpha \beta } &\doteqdot &\widehat{R}_{\ \alpha \beta \gamma
}^{\gamma },\ \ \widehat{R}_{\alpha \ ^{1}a}\doteqdot -\widehat{R}_{\ \alpha
\gamma \ ^{1}a}^{\gamma },\ \widehat{R}_{\ ^{1}a\alpha }\doteqdot \widehat{R}%
_{\ \ ^{1}a\alpha \ ^{1}b}^{\ ^{1}b},  \notag \\
&&\widehat{R}_{\ ^{1}a\ ^{1}b}\doteqdot \widehat{R}_{\ \ ^{1}a\ ^{1}b\
^{1}c}^{\ ^{1}c};  \notag
\end{eqnarray}%
\begin{eqnarray*}
\widehat{R}_{\ ^{1}\alpha \ ^{1}\beta } &\doteqdot &\widehat{R}_{\ \
^{1}\alpha \ ^{1}\beta \ ^{1}\gamma }^{\ ^{1}\gamma },\ \ \widehat{R}_{\
^{1}\alpha \ ^{2}a}\doteqdot -\widehat{R}_{\ \ ^{1}\alpha \ ^{1}\gamma \
^{2}a}^{\ ^{1}\gamma },  \notag \\
&&\widehat{R}_{\ ^{2}a\ ^{1}\alpha }\doteqdot \widehat{R}_{\ \ ^{2}a\
^{1}\alpha \ ^{2}b}^{\ ^{2}b},\ \widehat{R}_{\ ^{2}a\ ^{2}b}\doteqdot
\widehat{R}_{\ \ ^{2}a\ ^{2}b\ ^{2}c}^{\ ^{2}c};  \notag \\
&&........  \notag \\
\widehat{R}_{\ ^{k-1}\alpha \ ^{k-1}\beta } &\doteqdot &\widehat{R}_{\ \
^{k-1}\alpha \ ^{k-1}\beta \ ^{k-1}\gamma }^{\ ^{k-1}\gamma },\ \ \widehat{R}%
_{\ ^{k-1}\alpha \ ^{k}a}\doteqdot -\widehat{R}_{\ \ ^{k-1}\alpha \
^{k-1}\gamma \ ^{k}a}^{\ ^{k-1}\gamma },  \notag \\
\ \widehat{R}_{\ ^{k}a\ ^{k-1}\alpha } &\doteqdot &\widehat{R}_{\ \ ^{k}a\
^{k-1}\alpha \ ^{k}b}^{\ ^{k}b},\ \widehat{R}_{\ ^{k}a\ ^{k}b}\doteqdot
\widehat{R}_{\ \ ^{k}a\ ^{k}b\ ^{k}c}^{\ ^{k}c}.  \notag
\end{eqnarray*}
\end{corollary}

\begin{proof}
The formulas (\ref{dricci}) follow from contractions of (\ref{curvnc}). To
compute the N--adapted coefficients of the Ricci tensor $Ric(\widehat{%
\mathbf{D}})$ for metric $\ ^{k}\mathbf{g}$ (\ref{gendm}) we have to
construct correspondingly the formulas (\ref{dcurv}). $\square $
\end{proof}

\vskip5pt

\begin{definition}
The scalar curvature $\ ^{s}\widehat{R}$ of $\widehat{\mathbf{D}}$ is by
definition
\begin{eqnarray}
\ ^{s}\widehat{R} &\doteqdot &\mathbf{g}^{\ ^{k}\alpha \ ^{k}\beta }\widehat{%
\mathbf{R}}_{\ ^{k}\alpha \ ^{k}\beta }  \label{sdccurv} \\
&=&g^{ij}\widehat{R}_{ij}+h^{ab}\widehat{R}_{ab}+h^{\ ^{1}a\ ^{1}b}\widehat{R%
}_{\ ^{1}a\ ^{1}b}+...+h^{\ ^{k}a\ ^{k}b}\widehat{R}_{\ ^{k}a\ ^{k}b}.
\notag
\end{eqnarray}
\end{definition}

Using values (\ref{dricci}) and (\ref{sdccurv}), we can compute the Einstein
tensor $\widehat{\mathbf{E}}_{\ ^{k}\alpha \ ^{k}\beta }$ of $\widehat{%
\mathbf{D}},$
\begin{equation}
\widehat{\mathbf{E}}_{\ ^{k}\alpha \ ^{k}\beta }\doteqdot \widehat{\mathbf{R}%
}_{\ ^{k}\alpha \ ^{k}\beta }-\frac{1}{2}\mathbf{g}_{\ ^{k}\alpha \
^{k}\beta }\ ^{s}\widehat{R}.  \label{enstdt}
\end{equation}%
In explicit form, for N--adapted coefficients, we get a proof for

\begin{corollary}
The Einstein tensor $\widehat{\mathbf{E}}_{\ ^{k}\alpha \ ^{k}\beta }$
splits into h- and $\ ^{k}v$--components %
 $\widehat{\mathbf{E}}_{\ ^{k}\alpha \ ^{k}\beta } \doteqdot \{\widehat{E}%
_{ij}=\widehat{R}_{ij}-\frac{1}{2}g_{ij}\ ^{s}\widehat{R},\widehat{E}_{ia}=%
\widehat{R}_{ia},\widehat{E}_{ai}=\widehat{R}_{ai},\
\widehat{E}_{ab} =\widehat{R}_{ab}-\frac{1}{2}h_{ab}\ ^{s}\widehat{R};\
\widehat{E}_{\alpha \ ^{1}a} =\widehat{R}_{\alpha \ ^{1}a},\ \widehat{E}%
_{\ ^{1}a\beta }=\ \widehat{R}_{\ ^{1}a\beta },\
\widehat{E}_{\ ^{1}a\ ^{1}b} =\ \widehat{R}_{\ ^{1}a\ ^{1}b}-\frac{1}{2}%
h_{\ ^{1}a\ ^{1}b}\ ^{s}\widehat{R};\\
\widehat{E}_{\ ^{1}\alpha \ ^{2}a} =\widehat{R}_{\ ^{1}\alpha \ ^{2}a},\
\widehat{E}_{\ ^{2}a\ ^{1}\beta }=\ \widehat{R}_{\ ^{2}a\ ^{1}\beta },\
\widehat{E}_{\ ^{2}a\ ^{2}b} =\ \widehat{R}_{\ ^{2}a\ ^{2}b}-\frac{1}{2}%
h_{\ ^{2}a\ ^{2}b}\ ^{s}\widehat{R};\\
\widehat{E}_{\ ^{k-1}\alpha \ ^{k}a} =\widehat{R}_{\ ^{k-1}\alpha \ ^{k}a},\
\widehat{E}_{\ ^{k}a\ ^{k-1}\beta }=\ \widehat{R}_{\ ^{k}a\ ^{k-1}\beta },
 \widehat{E}_{\ ^{k}a\ ^{k}b} =\ \widehat{R}_{\ ^{k}a\ ^{k}b}\\
 -\frac{1}{2}%
h_{\ ^{k}a\ ^{k}b}\ ^{s}\widehat{R}\}.$
\end{corollary}

In different theories of (string/brane/gauge etc) gravity, we can consider
nonholonomically modified gravitational field equations
\begin{equation}
\widehat{\mathbf{E}}_{\ ^{k}\alpha \ ^{k}\beta }=\varkappa \widehat{\mathbf{T%
}}_{\ ^{k}\alpha \ ^{k}\beta },  \label{deinst}
\end{equation}%
for a source $\widehat{\mathbf{T}}_{\ ^{k}\alpha \ ^{k}\beta }$ defined by
certain classical or quantum corrections and/or constraints on dynamics of
fields to usual energy--momentum tensors. Such equations are not equivalent,
in general, to the usual Einstein equations (\ref{einsteq}) for the
Levi--Civita connection $\ ^{k}\nabla $ .\footnote{%
As we noted in Refs. \cite{vrflg,vbrane,vkvnsst,vdqak}, an equivalence of
both types of filed equations would be possible, for instance, if we
introduce a generalized source $\widehat{\mathbf{T}}_{\ ^{k}\beta \
^{k}\delta }$ containing contributions of the distortion tensor (\ref{deft}).%
}

\begin{condition}
A class of metrics $\ ^{k}\mathbf{g}$ (\ref{gendm}) defining solutions of
the gravitational field equations the canonical d--connection (\ref{deinst})
are also solutions for the Einstein equations for the Levi--Civita
connection (\ref{einsteq}), if with respect to certain N--adapted frames (%
\ref{dder}) and (\ref{ddif}) there are satisfied the conditions
\begin{eqnarray}
\widehat{C}_{jb}^{i} &=&0,\Omega _{\ ji}^{a}=0,\widehat{T}_{ja}^{c}=0;%
\widehat{C}_{\beta \ ^{1}b}^{\alpha }=0,\Omega _{\ \beta \alpha }^{\
^{1}a}=0,\widehat{T}_{\beta \ ^{1}a}^{\ ^{1}c}=0;  \notag \\
\widehat{C}_{\ ^{1}\beta \ ^{2}b}^{\ ^{1}\alpha } &=&0,\Omega _{\ \
^{1}\beta \ ^{1}\alpha }^{\ ^{2}a}=0,\widehat{T}_{\ ^{1}\beta \ ^{2}a}^{\
^{2}c}=0;...;  \label{lccondg} \\
\widehat{C}_{\ ^{k-1}\beta \ ^{k}b}^{\ ^{k-1}\alpha }&=&0,\Omega _{\ \
^{k-1}\beta \ ^{k-1}\alpha }^{\ ^{k}a}=0,\widehat{T}_{\ ^{k-1}\beta \
^{k}a}^{\ ^{k}c}=0;  \notag
\end{eqnarray}%
and $\varkappa \widehat{\mathbf{T}}_{\ ^{k}\alpha \ ^{k}\beta }$ includes
the energy--momentum tensor for matter field in usual gravity and
distortions of the Einstein tensor determined by distortions of linear
connections.
\end{condition}

\begin{proof}
We can see that both the torsion (\ref{dtor}) and distortion tensor, see
formulas (\ref{deft}), became zero if \ the conditions (\ref{lccondg}) are
satisfied. In such a case, the distortion relations (\ref{deflc}) transform
into $\ \Gamma _{\ \ ^{k}\alpha \ ^{k}\beta }^{\ ^{k}\gamma }=\widehat{%
\mathbf{\Gamma }}_{\ \ ^{k}\alpha \ ^{k}\beta }^{\ ^{k}\gamma }$ (even, in
general, $\widehat{\mathbf{D}}\neq \nabla ).$\footnote{%
This is possible because the laws of transforms for d--connections, for the
Levi--Civita connection and different types of tensors being adapted, or
not, to a N--splitting (\ref{whitney}) are very different.} Even such
additional constraints are imposed, the geometric constructions are with
nonholonomic variables because the anholonomy coefficients are not
obligatory zero(for instance, $w_{ia}^{b}=\partial _{a}N_{i}^{b}$ etc, see
formulas (\ref{anhrel})). $\square $
\end{proof}

\subsection{The system of N--adapted Einstein equations}

The goal of this work is to prove that we can solve in a very general form
any system of gravitational filed equations in high dimensional gravity, for
instance, for the canonical d--connection and/or the Levi--Civita connection
if such equations can be written as a variant of equations (\ref{deinst}),
\begin{equation}
\widehat{\mathbf{R}}_{\ \ \ ^{k}\beta }^{\ ^{k}\alpha }=\Upsilon _{\ \ \
^{k}\beta }^{\ ^{k}\alpha },  \label{deinst1}
\end{equation}%
with a general source parametrize in the form (\ref{source}), $\Upsilon _{\
\ \ ^{k}\beta }^{\ ^{k}\alpha }=diag[\Upsilon _{\ ^{k}\gamma }],$ including
possible contributions from energy--momentum and/or distortion tensors.

Let us denote partial derivatives in the form
\begin{equation*}
\partial _{2}=\partial /\partial x^{2},...,\partial _{v}=\partial /\partial
v,\partial _{^{k}v}=\partial /\partial \ ^{k}v,...,\partial _{\alpha
}=\partial /\partial u^{\alpha },...,\partial _{\ ^{k}\alpha }=\partial
/\partial u^{\ ^{k}\alpha }.
\end{equation*}

\begin{theorem}
\label{th2}The gravitational field equations (\ref{deinst1}) constructed for
$\ ^{k}\widehat{\mathbf{D}}=\{\widehat{\mathbf{\Gamma }}_{\ \ ^{k}\alpha \
^{k}\beta }^{\ ^{k}\gamma }\}$ with coefficients (\ref{candcon}) and
computed for a metric $\ ^{k}\mathbf{g=\mathbf{\{\mathbf{g}}}_{\ ^{k}\beta \
^{k}\gamma }\mathbf{\mathbf{\}}}$ (\ref{gendm}) with coefficients (\ref%
{data1}) are equivalent to this system of partial differential equations:
\begin{eqnarray}
\widehat{R}_{2}^{2} &=&\widehat{R}_{3}^{3} =\frac{1}{2g_{2}g_{3}}[\frac{%
\partial _{2}g_{2}\cdot \partial _{2}g_{3}}{2g_{2}}+\frac{(\partial
_{2}g_{3})^{2}}{2g_{3}}-\partial _{2}^{\ 2}g_{3}  \label{4ep1a} \\
&&+\frac{\partial _{3}g_{2}\cdot \partial _{3}g_{3}}{2g_{3}}+\frac{(\partial
_{3}g_{2})^{2}}{2g_{2}}-\partial _{3}^{\ 2}g_{2}]=-\Upsilon _{4}(x^{\widehat{%
i}}),  \notag \\
\widehat{R}_{4}^{4} &=&\widehat{R}_{5}^{5}=\frac{\partial _{v}h_{5}}{%
2h_{4}h_{5}}\partial _{v}\left( \ln \left| \frac{\sqrt{|h_{4}h_{5}|}}{%
\partial _{v}h_{5}}\right| \right) =\ -\Upsilon _{2}(x^{i},v),  \label{4ep2a}
\\
\widehat{R}_{4i} &=&-w_{i}\frac{\beta }{2h_{4}}-\frac{\alpha _{i}}{2h_{4}}=0,
\label{4ep3a} \\
\widehat{R}_{5i} &=&-\frac{h_{5}}{2h_{4}}\left[ \partial _{v}^{\
2}n_{i}+\gamma \partial _{v}n_{i}\right] =0,  \label{4ep4a}
\end{eqnarray}%
\begin{eqnarray*}
\widehat{R}_{6}^{6} &=&\widehat{R}_{7}^{7}=\frac{\partial _{\ ^{1}v}h_{7}}{%
2h_{6}h_{7}}\partial _{\ ^{1}v}\left( \ln \left| \frac{\sqrt{|h_{6}h_{7}|}}{%
\partial _{\ ^{1}v}h_{6}}\right| \right) =\ -\ ^{1}\Upsilon _{2}(u^{\alpha
},\ ^{1}v), \\
\widehat{R}_{6\mu } &=&-w_{\mu }\frac{\ ^{1}\beta }{2h_{6}}-\frac{\alpha
_{\mu }}{2h_{6}}=0,\
\widehat{R}_{7\mu } = -\frac{h_{7}}{2h_{6}}\left[ \partial _{\ ^{1}v}^{\
2}n_{\mu }+\ ^{1}\gamma \partial _{\ ^{1}v}n_{\mu }\right] =0,
\end{eqnarray*}%
\begin{eqnarray*}
\widehat{R}_{8}^{8} &=&\widehat{R}_{9}^{9}=\frac{\partial _{\ ^{2}v}h_{9}}{%
2h_{8}h_{9}}\partial _{\ ^{2}v}\left( \ln \left| \frac{\sqrt{|h_{8}h_{9}|}}{%
\partial _{\ ^{2}v}h_{8}}\right| \right) =\ -\ ^{2}\Upsilon _{2}(u^{\
^{1}\alpha },\ ^{2}v), \\
\widehat{R}_{8\ ^{1}\mu } &=&-w_{\ ^{1}\mu }\frac{\ ^{2}\beta }{2h_{8}}-%
\frac{\alpha _{\ ^{1}\mu }}{2h_{8}}=0,\\
\widehat{R}_{9\ ^{1}\mu } &=& -\frac{h_{9}}{2h_{8}}\left[ \partial _{\
^{2}v}^{\ 2}n_{\ ^{1}\mu }+\ ^{2}\gamma \partial _{\ ^{2}v}n_{\ ^{1}\mu }%
\right] =0,
\end{eqnarray*}%
\begin{equation*}
.....,
\end{equation*}%
\begin{eqnarray*}
\widehat{R}_{4+2k}^{4+2k} &=&\widehat{R}_{5+2k}^{5+2k}=\frac{\partial _{\
^{k}v}h_{5+2k}}{2h_{4+2k}h_{5+2k}}\partial _{\ ^{k}v}\left( \ln \left| \frac{%
\sqrt{|h_{4+2k}h_{5+2k}|}}{\ }{\partial _{\ ^{k}v}h_{4+2k}}\right| \right) \\
&=&\ -\ ^{k}\Upsilon _{2}(u^{\ ^{k-1}\alpha },\ ^{k}v),
\end{eqnarray*}%
\begin{eqnarray*}
\widehat{R}_{4+2k\ ^{k-1}\mu } &=&-w_{\ ^{k-1}\mu }\frac{\ ^{k}\beta }{%
2h_{4+2k}}-\frac{\alpha _{\ ^{k-1}\mu }}{2h_{4+2k}}=0, \\
\widehat{R}_{5+2k\ ^{k-1}\mu } &=&-\frac{h_{5+2k}}{2h_{4+2k}}\left[ \partial
_{\ ^{k}v}^{\ 2}n_{\ ^{k-1}\mu }+\ ^{k}\gamma \partial _{\ ^{k}v}n_{\
^{k-1}\mu }\right] =0,
\end{eqnarray*}%
where, for $\partial _{v}h_{4}\neq 0$ and $\partial _{v}h_{5}\neq 0;\partial
_{^{1}v}h_{6}\neq 0$ and $\partial _{^{1}v}h_{7}\neq 0;\partial
_{^{2}v}h_{8}\neq 0$ and $\partial _{^{2}v}h_{9}\neq 0;...;\partial
_{^{k}v}h_{4+2k}\neq 0$ and $\partial _{^{k}v}h_{5+2k}\neq 0;$\footnote{%
solutions, for instance, with $\partial _{v}h_{4}=0$ and/or $\partial
_{v}h_{5}=0,$ should be analyzed as some special cases (for simplicity, we
omit such considerations in this work)}%
\begin{eqnarray}
~\phi &=&\ln |\frac{\partial _{v}h_{5}}{\sqrt{|h_{4}h_{5}|}}|,\ \alpha
_{i}=\partial _{v}h_{5}\cdot \partial _{i}\phi ,\   \label{auxphi} \\
\beta &=&\partial _{v}h_{4}\ \cdot \partial _{v}\phi ,\ \gamma =\partial
_{v}\left( \ln |h_{5}|^{3/2}/|h_{4}|\right) ;  \notag
\end{eqnarray}%
\begin{eqnarray*}
~\ ^{1}\phi &=&\ln |\frac{\partial _{\ ^{1}v}h_{7}}{\sqrt{|h_{6}h_{7}|}}|,\
\alpha _{\mu }=\partial _{\ ^{1}v}h_{7}\ \cdot \partial _{\mu }\ ^{1}\phi ,\
\\
\ ^{1}\beta &=&\partial _{\ ^{1}v}h_{6}\ \cdot \partial _{\ ^{1}v}\ ^{1}\phi
,\ \ ^{1}\gamma =\partial _{\ ^{1}v}\left( \ln |h_{7}|^{3/2}/|h_{6}|\right) ;
\end{eqnarray*}%
\begin{eqnarray*}
~\ ^{2}\phi &=&\ln |\frac{\partial _{\ ^{2}v}h_{9}}{\sqrt{|h_{8}h_{9}|}}|,\
\alpha _{\ ^{1}\mu }=\partial _{\ ^{2}v}h_{9}\ \cdot \partial _{\ ^{1}\mu }\
^{2}\phi ,\  \\
\ ^{2}\beta &=&\partial _{\ ^{2}v}h_{8}\ \cdot \partial _{\ ^{2}v}\ ^{2}\phi
,\ \ ^{2}\gamma =\partial _{\ ^{2}v}\left( \ln |h_{9}|^{3/2}/|h_{8}|\right) ;
\end{eqnarray*}%
\begin{equation*}
......;
\end{equation*}%
\begin{eqnarray*}
~\ ^{k}\phi &=&\ln |\frac{\partial _{\ ^{k}v}h_{5+2k}}{\sqrt{%
|h_{4+2k}h_{5+2k}|}}|,\ \alpha _{\ ^{k-1}\mu }=\partial _{\ ^{k}v}h_{5+2k}\
\cdot \partial _{\ ^{k-1}\mu }\ ^{k}\phi ,\  \\
\ ^{k}\beta &=&\partial _{\ ^{k}v}h_{4+2k}\ \cdot \partial _{\ ^{k}v}\
^{k}\phi ,\ \ ^{k}\gamma =\partial _{\ ^{k}v}\left( \ln
|h_{5+2k}|^{3/2}/|h_{4+2k}|\right) .
\end{eqnarray*}
\end{theorem}

Proof of this theorem is sketched in section \ref{proofth2}.

\vskip5pt Finally, we emphasize that the system of equations constructed in
Theorem \ref{th2} can be integrated in very general forms. For instance, for
any given $\Upsilon _{4}$ and $\Upsilon _{4},$ the equation (\ref{4ep1a})
relates an un--known function $g_{2}(x^{2},x^{3})$ to a prescribed $%
g_{3}(x^{2},x^{3}),$ or inversely. The equation (\ref{4ep2a}) contains only
derivatives on $y^{4}=v$ and allows us to define $h_{4}(x^{i},v)$ for a
given $h_{5}(x^{i},v),$ or inversely, for $h_{4,5}^{\ast }\neq 0.$ Having
defined $h_{4}$ and $h_{5},$ we can compute the coefficients (\ref{auxphi}),
which allows us to find $w_{i}$ from algebraic equations (\ref{4ep3a}) and
to compute $n_{i}$ $\ $by integrating two times on $v$ as follow from
equations (\ref{4ep4a}). Similar properties hold true for equations on
higher order shells.

\section{General Solutions for Einstein Equations with Extra Dimensions}

\label{s4}In this section, we show how general solutions of the
gravitational field equations can be constructed in explicit form. There are
three key steps: The first one is to generate exact solutions with Killing
symmetries for the canonical d--connection. At the second one, we shall
analyze the constraints selecting solutions for the Levi--Civita
connections. The final (third) step will be in generalizing the
constructions by eliminating Killing symmetries.

\subsection{Exact solutions with Killing symmetries}

We formulate for the Einstein equations for the canonical d--connection:

\begin{theorem}
\label{th3}The general class of solutions of nonholonomic gravitational
equations (\ref{deinst1}) with Killing symmetries on $e_{5+2k}=\partial
/\partial y^{5+2k}$ is defined by ansatz of type (\ref{ansgensol}) with $\
^{k}\omega ^{2}=1$ and coefficients $g_{\widehat{i}},h_{\ ^{k}a},$ $w_{\
^{k-1}\alpha },n_{\ ^{k-1}\alpha }$ computed for $k=0,1,2,...$ following
formulas (\ref{coeff}).
\end{theorem}

\begin{proof}
We sketch the proof giving more details for the shell $k=0$ (higher order
constructions being similar):
\begin{itemize}
\item The general solution of equation (\ref{4ep1a}) can be written in the form
$ \varpi =g_{[0]}\exp [a_{2}\widetilde{x}^{2}\left( x^{2},x^{3}\right) +a_{3}%
\widetilde{x}^{3}\left( x^{2},x^{3}\right) ],$
were $g_{[0]},a_{2}$ and $a_{3}$ are some constants and the functions $%
\widetilde{x}^{2,3}\left( x^{2},x^{3}\right) $ define any coordinate
transforms $x^{2,3}\rightarrow \widetilde{x}^{2,3}$ for which the 2D line
element becomes conformally flat, i. e.
\begin{equation*}
g_{2}(x^{2},x^{3})(dx^{2})^{2}+g_{3}(x^{2},x^{3})(dx^{3})^{2}\rightarrow
\varpi (x^{2},x^{3})\left[ (d\widetilde{x}^{2})^{2}+\epsilon (d\widetilde{x}%
^{3})^{2}\right] ,  \label{con10}
\end{equation*}%
where $\epsilon =\pm 1$ for a corresponding signature. It is convenient to
write some partial derivatives, in brief, in the form $\partial
_{2}g=g^{\bullet },\partial _{3}g=g^{^{\prime }},\partial _{4}g=g^{\ast }.$
In coordinates $\widetilde{x}^{2,3},$ the equation (\ref{4ep1a}) transform
into %
 $\varpi \left( \ddot{\varpi}+\varpi ^{\prime \prime }\right) -\dot{\varpi}%
-\varpi ^{\prime }=2\varpi ^{2}\Upsilon _{4}(\tilde{x}^{2},\tilde{x}^{3})$
 or%
\begin{equation}
\ddot{\psi}+\psi ^{\prime \prime }=2\Upsilon _{4}(\tilde{x}^{2},\tilde{x}%
^{3}),  \label{auxeq01}
\end{equation}%
for $\psi =\ln |\varpi |.$ The integrals of (\ref{auxeq01}) depends on the
source $\Upsilon _{4}.$ As a particular case we can consider that $\Upsilon
_{4}=0.$

\item For $\ ^{k}\Upsilon _{2}(u^{\ ^{k-1}\alpha },\ ^{k}v)=0,$ the equation
(\ref{4ep2a}), and its higher shell analogs, relates two functions $%
h_{4+2k}(u^{\ ^{k-1}\alpha },\ ^{k}v)$ and $h_{5+2k}(u^{\ ^{k-1}\alpha },\
^{k}v)$ following two possibilities:

a) to compute
\begin{eqnarray}
&&\sqrt{|h_{5+2k}|} =\ _{1}h_{5+2k}\left( u^{\ ^{k-1}\alpha }\right) +\
_{2}h_{5+2k}\left( u^{\ ^{k-1}\alpha }\right) \times  \notag \\
&& \int \sqrt{|h_{4+2k}(u^{\ ^{k-1}\alpha },\ ^{k}v)|}dv,  \mbox{\ for \ } \partial _{\ ^{k}v}h_{4+2k}(u^{\ ^{k-1}\alpha },\ ^{k}v)
\neq 0;  \notag \\
&&=\ _{1}h_{5+2k}\left( u^{\ ^{k-1}\alpha }\right) +\ _{2}h_{5+2k}\left(
u^{\ ^{k-1}\alpha }\right) \ ^{k}v,  \label{p2} \\
&& \mbox{\ for \ } \partial _{\ ^{k}v}h_{4+2k}(u^{\ ^{k-1}\alpha },\ ^{k}v)
=0,  \notag
\end{eqnarray}%
for some functions $\ _{1}h_{5+2k}\left( u^{\ ^{k-1}\alpha }\right) $ and $\
_{2}h_{5+2k}\left( u^{\ ^{k-1}\alpha }\right) $ stated by boundary
conditions;

b) or, inversely, to compute $h_{4+2k}$ for respectively given $h_{5+2k},$
with $~\partial _{\ ^{k}v}h_{5+2k}\neq 0,$%
\begin{equation}
\sqrt{|h_{4+2k}|}=\ _{k}^{0}h\left( u^{\ ^{k-1}\alpha }\right) ~\partial _{\
^{k}v}\sqrt{|h_{5+2k}\left( u^{\ ^{k-1}\alpha },\ ^{k}v\right) |},
\label{p1}
\end{equation}%
with $\ _{k}^{0}h\left( u^{\ ^{k-1}\alpha }\right) $ given by boundary
conditions. We note that the (\ref{4ep2a}) with zero source is satisfied by
arbitrary pairs of coefficients $h_{4+2k}(u^{\ ^{k-1}\alpha },\ ^{k}v)$ and $%
\ _{0}h_{5+2k}\left( u^{\ ^{k-1}\alpha }\right) .$ Solutions with $\
^{k}\Upsilon _{2}\neq 0$ can be found by ansatz of type
\begin{equation}
h_{5+2k}[\ ^{k}\Upsilon _{2}]=h_{5+2k},h_{4}[\ ^{k}\Upsilon _{2}]=\varsigma
_{4+2k}\left( u^{\ ^{k-1}\alpha },\ ^{k}v\right) h_{4+2k},  \label{aux2f}
\end{equation}%
where $h_{4+2k}$ and $h_{5+2k}$ are related by formula (\ref{p2}), or (\ref%
{p1}). Substituting (\ref{aux2f}), we obtain%
\begin{equation}
\varsigma _{4+2k}\left( u^{\ ^{k-1}\alpha },\ ^{k}v\right) =\ ^{0}\varsigma
_{4+2k}\left( u^{\ ^{k-1}\alpha }\right) -\int \ ^{k}\Upsilon _{2} \frac{%
h_{4+2k}h_{5+2k}}{4\partial _{\ ^{k}v}h_{5+2k}}d\ ^{k}v,  \label{aux3f}
\end{equation}%
where $\ ^{0}\varsigma _{4+2k}\left( u^{\ ^{k-1}\alpha }\right) $ are
arbitrary functions.

\item The exact solutions of (\ref{4ep3a}) for $\beta \neq 0$ are defined
from an algebraic equation, $w_{i}\beta +\alpha _{i}=0,$ where the
coefficients $\beta $ and $\alpha _{i}$ are computed as in formulas (\ref%
{auxphi}) by using the solutions for (\ref{4ep1a}) and (\ref{4ep2a}). The
general solution is
\begin{equation}
w_{k}=\partial _{k}\ln [\sqrt{|h_{4}h_{5}|}/|h_{5}^{\ast }|]/\partial
_{v}\ln [\sqrt{|h_{4}h_{5}|}/|h_{5}^{\ast }|],  \label{w1}
\end{equation}%
with $\partial _{v}=\partial /\partial v$ and $h_{5}^{\ast }\neq 0.$ If $%
h_{5}^{\ast }=0,$ or even $h_{5}^{\ast }\neq 0$ but $\beta =0,$ the
coefficients $w_{k}$ could be arbitrary functions on $\left( x^{i},v\right)
. $ \ For the vacuum Einstein equations this is a degenerated case imposing
the the compatibility conditions $\beta =\alpha _{i}=0,$ which are
satisfied, for instance, if the $h_{4}$ and $h_{5}$ are related as in the
formula (\ref{p1}) but with $h_{[0]}\left( x^{i}\right) =const.$

\item Having defined $h_{4}$ and $h_{5}$ and computed $\gamma $ from (\ref%
{auxphi}), we can solve the equation (\ref{4ep4a}) by integrating on
variable ''$v$'' the equation $n_{i}^{\ast \ast }+\gamma n_{i}^{\ast }=0.$
The exact solution is
\begin{eqnarray}
n_{k} &=&n_{k[1]}\left( x^{i}\right) +n_{k[2]}\left( x^{i}\right) \int
[h_{4}/(\sqrt{|h_{5}|})^{3}]dv,~h_{5}^{\ast }\neq 0;  \notag \\
&=&n_{k[1]}\left( x^{i}\right) +n_{k[2]}\left( x^{i}\right) \int
h_{4}dv,\qquad ~h_{5}^{\ast }=0;  \label{n1} \\
&=&n_{k[1]}\left( x^{i}\right) +n_{k[2]}\left( x^{i}\right) \int [1/(\sqrt{%
|h_{5}|})^{3}]dv,~h_{4}^{\ast }=0,  \notag
\end{eqnarray}%
for some functions $n_{k[1,2]}\left( x^{i}\right) $ stated by boundary
conditions.

\item The generating and integration formulas in higher order formulas (\ref%
{aux2f}), (\ref{aux3f}), (\ref{w1}), (\ref{n1}) etc are redefined in a form
as it was considered for $k=0$ in review articles \cite{vrflg,ijgmmp} which
result in formulas (\ref{coeff}) for $\ ^{k}\omega ^{2}=1.\ \square $
\end{itemize}
\end{proof}

\vskip5pt

We note that the solutions constructed in Theorem \ref{th3} are very general
ones and contain as particular cases all known exact solutions for (non)
holonomic Einstein spaces with Killing symmetries. They also can be
generalized to include arbitrary finite sets of parameters, see Ref. \cite%
{ijgmmp}.

\begin{corollary}
\label{cor1}An ansatz (\ref{ansgensol}) with $\ ^{k}\omega ^{2}=1$ and
coefficients $g_{\widehat{i}},h_{\ ^{k}a},$ $w_{\ ^{k-1}\beta },$ $n_{\
^{k-1}\beta }$ computed following formulas (\ref{coeff}) define solutions
with Killing symmetries on $e_{5+2k}=\partial /\partial y^{5+2k}$ of the
Einstein equations (\ref{einst1}) for the Levi--Civita connection $\Gamma
_{\ \ ^{k}\alpha \ ^{k}\beta }^{\ ^{k}\gamma }$ if the coefficients of
metric are subjected additionally to the conditions (\ref{lccond}).
\end{corollary}

\begin{proof}
By straightforward computations for ansatz (\ref{gendm}) with coefficients (%
\ref{data1}), we get that the conditions (\ref{lccondg}) resulting in $%
\Gamma _{\ \ ^{k}\alpha \ ^{k}\beta }^{\ ^{k}\gamma }=\widehat{\mathbf{%
\Gamma }}_{\ \ ^{k}\alpha \ ^{k}\beta }^{\ ^{k}\gamma }$ are just those
written as (\ref{lccond}).\ For such ansatz, one $\ $holds the conditions (%
\ref{aux02cc})$\ $\ and the N--connection and torsion coefficients vanish,
i. e. the values (\ref{omeg}) and (\ref{dtorsc}) became zero. We get
nonholonomic configurations for the Levi--Civita connection with nontrivial
anholonomy coefficients (\ref{anhrel}). $\square $
\end{proof}

\vskip5pt

We conclude that in order to generate exact solutions with Killing
symmetries in Einstein gravity and its higher order generalizations, we
should consider N--adapted frames and nonholonomic deformations of the
Levi--Civita connection to an auxiliary metric compatible d--connection (for
instance, to the canonical d--connection, $\ ^{k}\widehat{\mathbf{D}}),$
when the corresponding system of nonholonomic gravitational field equations (%
\ref{4ep1a})--(\ref{4ep4a}) can be integrated in general form. Subjecting
the integral variety of such solutions to additional constraints of type (%
\ref{lccondg}), i.e. imposing the conditions (\ref{lccond}) to the
coefficients of metrics, we may construct new classes of exact solutions of
Einstein equations for the Levi--Civita connection $\ ^{k}\nabla .$

\subsection{General non--Killing solutions}

Our final aim is to consider general classes of solutions of the
nonholonomic gravitational field equations (\ref{deinst1}), and (for more
particular cases), of Einstein equations (\ref{einst1}) metrics depending on
all coordinates $u^{\ ^{k}\alpha }=(x^{i},y^{\ ^{k}a}),$ i.e. the solutions
will be without Killing symmetries.

Let us introduce some nontrivial multiples $\ ^{k}\omega ^{2}(u^{\ ^{k}\alpha })$
before coefficients $h_{\ ^{k}a}$ parametrized in the form (\ref{data1}) and
defining solutions with Killing symmetries. We get an ansatz \
\begin{eqnarray}
\ _{\omega }^{k}\mathbf{g} &=&\epsilon _{1}e^{1}\otimes e^{1}+g_{\widehat{j}%
}(x^{\widehat{k}})e^{\widehat{j}}\otimes e^{\widehat{j}}+\omega
^{2}(x^{i},y^{a})h_{a}(x^{i},v)\mathbf{e}^{a}\otimes \mathbf{e}^{a}  \notag
\\
&&+\ ^{1}\omega ^{2}(u^{\ ^{1}\alpha })\ \ h_{\ ^{1}a\ ^{1}b}(u^{\ \alpha
},\ ^{1}v)\mathbf{e}^{\ ^{1}a}\otimes \mathbf{e}^{\ ^{1}b}  \notag \\
&& +\ ^{2}\omega ^{2}(u^{\ ^{2}\alpha })h_{\ ^{2}a\ ^{2}b}(u^{\ ^{1}\alpha
},\ ^{2}v)\mathbf{e}^{\ ^{2}a}\otimes \mathbf{e}^{\ ^{2}b} +\ldots  \notag \\
&& +\ ^{k}\omega ^{2}(u^{\ ^{k}\alpha })\ h_{\ ^{k}a\ ^{k}b}(u^{\
^{k-1}\alpha },\ ^{k}v)\mathbf{e}^{\ ^{k}a}\otimes \mathbf{e}^{\ ^{k}b},
\label{genansc}
\end{eqnarray}%
where the N--adapted basis $\mathbf{e}^{\ ^{k}a}$ (and N--connection) are
the same as in (\ref{gendm}). Under such noholonomic conformal transform%
\footnote{%
we use the term ''nonholonomic'' because such transforms/ deformations are
adapted to a N--splitting stated by a prescribed nonholonomic distribution
on a corresponding high dimension spacetime} (defined by generating
functions $\ ^{2}\omega ^{2}(u^{\ ^{k}\alpha }))$ of metric, $\ ^{k}\mathbf{%
g\rightarrow \ }\ _{\ ^{k}\omega }^{\ }\mathbf{g,}$ the canonical
d--connection deforms as $\widehat{\mathbf{\Gamma }}_{\ \ ^{k}\alpha \
^{k}\beta }^{\ ^{k}\gamma }\rightarrow $ $\ _{\ ^{k}\omega }^{\ }\widehat{%
\mathbf{\Gamma }}_{\ \ ^{k}\alpha \ ^{k}\beta }^{\ ^{k}\gamma },$ where
 $\ _{\ ^{k}\omega }^{\ }\widehat{\mathbf{\Gamma }}_{\ \ ^{k}\alpha \
^{k}\beta }^{\ ^{k}\gamma } = (\widehat{L}_{jk}^{i},_{\ \omega }^{\ }%
\widehat{L}_{bk}^{a},\widehat{C}_{jc}^{i},_{\ \omega }^{\ }\widehat{C}%
_{bc}^{a};$ $_{\ ^{1}\omega }^{\ }\widehat{L}_{\ ^{1}b\alpha }^{\ ^{1}a},_{\
\omega }^{\ }\widehat{C}_{\beta \ ^{1}c}^{\alpha },_{\ ^{1}\omega }^{\ }%
\widehat{C}_{\ ^{1}b\ ^{1}c}^{\ ^{1}a};\ldots ;\
 _{\ ^{k}\omega }^{\ }\widehat{L}_{\ ^{k}b\ ^{k-1}\alpha }^{\ ^{k}a},_{\
^{k-1}\omega }^{\ }\widehat{C}_{\ ^{k-1}\beta \ ^{k}c}^{\ ^{k-1}\alpha },_{\
^{k}\omega }^{\ }\widehat{C}_{\ ^{k}b\ ^{k}c}^{\ ^{k}a}),$
\begin{eqnarray*}
&&\mbox{ for } \widehat{C}_{jc}^{i} = _{\ \omega }^{\ }\widehat{C}_{\beta \ ^{1}c}^{\alpha
}=...=\ _{\ ^{k-1}\omega }^{\ }\widehat{C}_{\ ^{k-1}\beta \ ^{k}c}^{\
^{k-1}\alpha }=0,\  \ _{\ ^{k}\omega }^{\ }\widehat{L}_{b\ ^{k-1}\alpha }^{\ ^{k}a}
= \\
&&\widehat{L}%
_{\ ^{k}b\ ^{k-1}\alpha }^{\ ^{k}a}+\ _{\ ^{k}\omega }^{\ z}\widehat{L}_{\
^{k}b\ ^{k-1}\alpha }^{\ ^{k}a},
\ _{\ ^{k}\omega }^{\ }\widehat{C}_{\ ^{k}b\ ^{k}c}^{\ ^{k}a} = \widehat{C}%
_{\ ^{k}b\ ^{k}c}^{\ ^{k}a}+\ \ _{\ ^{k}\omega }^{\ z}\widehat{C}_{\ ^{k}b\
^{k}c}^{\ ^{k}a},
\end{eqnarray*}%
\begin{eqnarray}
&&\mbox{with } \ _{\ ^{k}\omega }^{\ z}
\widehat{L}_{\ ^{k}b\ ^{k-1}\alpha }^{\ ^{k}a} =%
\frac{1}{2\ ^{k}\omega ^{2}}h^{\ ^{k}a\ ^{k}c}\ [ h_{\ ^{k}b\ ^{k}c}\mathbf{e%
}_{\ ^{k-1}\beta }(\ ^{k}\omega ^{2})  \notag \\
&& -h_{\ ^{k}b\ ^{k}c}N_{\ ^{k-1}\beta }^{\ ^{k}d}\partial _{\ ^{k}d}(\
^{k}\omega ^{2})]  =\delta _{\ ^{k}b}^{^{k}a}\mathbf{e}_{\ ^{k-1}\beta }\ln |\ ^{k}\omega |;
\label{defcdc1} \\
 &&\ _{\ ^{k}\omega }^{\ z}\widehat{C}_{\ ^{k}b\ ^{k}c}^{\ ^{k}a} =\left(
\delta _{\ ^{k}b}^{^{k}a}\partial _{\ ^{k}c}\mathbf{+}\delta _{\
^{k}c}^{^{k}a}\partial _{\ ^{k}b}-h_{\ ^{k}b\ ^{k}e}h^{\ ^{k}a\
^{k}e}\partial _{\ ^{k}e}\right) \ln |\ ^{k}\omega |,  \label{defcdc2}
\end{eqnarray}%
are computed by introducing coefficients of $\ _{\omega }^{k}\mathbf{g}$ (%
\ref{genansc}) into (\ref{candcon}).

\begin{proposition}
For nonholonomic N--adapted transforms $\ ^{k}\mathbf{g}$ $\ $(\ref{gendm}) $%
\mathbf{\rightarrow \ }\ _{\ ^{k}\omega }^{\ }\mathbf{g}$ (\ref{genansc})
with shell coefficients $\ ^{k}\omega $ satisfying respectively the
conditions $\mathbf{e}_{\ ^{k-1}\beta }(\ ^{k}\omega)$ $=0,$ the Ricci
tensor transform $\widehat{\mathbf{R}}_{\ ^{k}\alpha \ ^{k}\beta }$ (\ref%
{driccia}) $\rightarrow \ _{\ ^{k}\omega }^{\ \ }\widehat{\mathbf{R}}_{\
^{k}\alpha \ ^{k}\beta },$ where
\begin{eqnarray}
\mathbf{\ }\ _{\ ^{k}\omega }^{\ }\widehat{\mathbf{R}}_{\ ^{k}\alpha \
^{k}\beta } &=&\{\widehat{R}_{ij},\widehat{R}_{ia},\ \widehat{R}_{ai},_{\
\omega }^{\ }\widehat{R}_{ab};_{\ \omega }^{\ }\widehat{R}_{\alpha \beta },%
\widehat{R}_{\alpha \ ^{1}a},\ \widehat{R}_{\ ^{1}a\beta },_{\ ^{1}\omega
}^{\ }\widehat{R}_{\ ^{1}a\ ^{1}b};  \notag \\
&&\ _{\ ^{1}\omega }^{\ }\widehat{R}_{\ ^{1}\alpha \ ^{1}\beta },\widehat{R}%
_{\ ^{1}\alpha \ ^{2}a},\ \widehat{R}_{\ ^{2}a\ ^{1}\beta },\ _{\ ^{2}\omega
}^{\ }\widehat{R}_{\ ^{2}a\ ^{2}b};...;  \label{dricciac} \\
&&\mathbf{\ }\ _{\ ^{k-1}\omega }^{\ }\widehat{R}_{\ ^{k-1}\alpha \
^{k-1}\beta },\widehat{R}_{\ ^{k-1}\alpha \ ^{k}a},\ \widehat{R}_{\ ^{k}a\
^{k-1}\beta },\ _{\ ^{k}\omega }^{\ }\widehat{R}_{\ ^{k}a\ ^{k}b}\},  \notag
\end{eqnarray}%
with $_{\ ^{k-1}\omega }^{\ }\widehat{R}_{\ ^{k-1}\alpha \ ^{k-1}\beta }$
computed recurrently using $\mathbf{\ }\widehat{R}_{\ ^{k-1}\alpha \
^{k-1}\beta }$ and \\
 $\ _{\ ^{k}\omega }^{\ }\widehat{R}_{\ ^{k}a\ ^{k}b}=\ \widehat{R}_{\ ^{k}a\
^{k}b}+\ _{\ ^{k}\omega }^{\ z}\widehat{R}_{\ ^{k}a\ ^{k}b},$
 where the deformation tensor $\ _{\ ^{k}\omega }^{\ z}\widehat{R}_{\ ^{k}a\
^{k}b}$ is given by formula
\begin{eqnarray}
&&\ _{\ ^{k}\omega }^{\ z}\widehat{R}_{\ ^{k}a\ ^{k}b}=(2-\ ^{k}m)\
\widehat{D}_{\ ^{k}a}\widehat{D}_{\ ^{k}b}\ln |\ ^{k}\omega |
 -h_{\ ^{k}a\ ^{k}b}h^{\ ^{k}c\ ^{k}d}\times \notag \\ &&  \widehat{D}_{\ ^{k}c}\widehat{D}_{\
^{k}d}\ln |\ ^{k}\omega |
-(2-\ ^{k}m)\ \left( \widehat{D}_{\ ^{k}a}\ln |\ ^{k}\omega |\right)
\widehat{D}_{\ ^{k}b}\ln |\ ^{k}\omega |  \label{riccivc} \\
&&+(2-\ ^{k}m)\ h_{\ ^{k}a\ ^{k}b}h^{\ ^{k}c\ ^{k}d}\ \left( \widehat{D}_{\
^{k}c}\ln |\ ^{k}\omega |\right) \widehat{D}_{\ ^{k}d}\ln |\ ^{k}\omega |.
\notag
\end{eqnarray}
\end{proposition}

\begin{proof}
It follows from an explicit computation of N--adapted coefficients of (\ref%
{dricciac}) taking into account the deformation relations (\ref{defcdc1})
and (\ref{defcdc2}) when the condition $\mathbf{e}_{\ ^{k-1}\beta }(\
^{k}\omega )=0,$ which is just (\ref{confcondk}) from the Main Theorem \ref%
{mth}. Working with shell coordinates $y^{\ ^{k}a},$ the formulas for
curvature and Ricci tensors are the same as on usual (pseudo) Riemannian
spaces for the Levi--Civita connection, when coordinates of type $x^{i}$ and
$y^{\ ^{k-1}a}$ can be considered as some parameters. For ''pure vertical ''
components, we can apply usual formulas for conformal transforms, like (\ref%
{defcdc2}) and (\ref{riccivc}) outlined, for instance, in Appendix D of
monograph \cite{wald}.\ $\square $
\end{proof}

\vskip5pt

\begin{remark}
\begin{enumerate}
\item There are two reasons to consider two dimensional shells with $\
^{k}m=2:$ \ The first one is that this results in field equations of type (%
\ref{4ep2a}) which can be integrated in general form (it is a problem, at
least technically, to find exact solutions for $^{k}m>2).$ The second one is
that from (\ref{riccivc}) we get $\ _{\ ^{k}\omega }^{\ z}\widehat{R}_{\ \
^{k}b}^{^{k}a}=\delta _{\ \ ^{k}b}^{^{k}a}\ ^{k}\widehat{\square }\ln |\
^{k}\omega |,$ with $\ ^{k}\widehat{\square }\doteqdot h^{\ ^{k}c\ ^{k}d}%
\widehat{D}_{\ ^{k}c}\widehat{D}_{\ ^{k}d}$ being a shell type d'Alambert
operator defined by the canonical d--connection.

\item We can impose additionally the conditions
\begin{equation}
\ ^{k}\widehat{\square }\ln |\ ^{k}\omega |=0,  \label{confd}
\end{equation}%
or to include such terms in sources (\ref{source}), redefining the
nonholonomic distributions to have
\begin{equation}
\ ^{k}\Upsilon _{2}(u^{\ ^{k-1}\alpha },\ ^{k}v)=\ _{\omega }^{k}\Upsilon
_{2}(u^{\ \ ^{k}\alpha })-\ ^{k}\widehat{\square }\ln |\ ^{k}\omega (u^{\
^{k}\alpha })|,  \label{sourscd}
\end{equation}%
for some well defined $\ _{\omega }^{k}\Upsilon _{2}(u^{\ \ ^{k}\alpha }) $
when formulas of type (\ref{aux3f}) can be computed. The conditions (\ref%
{confd}) or (\ref{sourscd}) can be selected also by corresponding
integration functions, for instance, in (\ref{p1}), (\ref{w1}) and/or (\ref%
{n1}) and/or their higher shell analogs.

\item Solutions with $^{k}m>2$ can be with different topologies and
generalized nonholonomic conformal symmetries. Locally such constructions
may be performed in a simplest way by considering formulas only with $%
^{k}m=2 $ by increasing the number of ''formal'' shells.
\end{enumerate}
\end{remark}

As a result, we get the proof of

\begin{lemma}
\label{lm1}Any metric parametrized in the form (\ref{genansc}) with
coefficients depending on all variables on a (pseudo) Riemannian manifold $\
^{k}\mathbf{V}$ ($\dim $ $\ ^{k}\mathbf{V}=3,$ or $2,+2k;\ $\ with $%
k=0,1,2,...$ two dimensional shells) defines a ''non--Killing'' solution of
the Einstein equations for the canonical d--connection $\ ^{k}\widehat{%
\mathbf{D}}$ if the coefficients are given by data (\ref{deinst1}) as
solutions with Killing symmetries of (\ref{4ep1a})--(\ref{4ep4a}), when the
parameters of nonholonomic conformal deformations $^{k}\omega (u^{\
^{k}\alpha })$ are chosen as generating functions satisfying the conditions $%
\mathbf{e}_{\ ^{k-1}\beta }(\ ^{k}\omega )=0$ and, for instance, $\ ^{k}%
\widehat{\square }\ln |\ ^{k}\omega |=0.$ Imposing additional restrictions
on integration functions as in Corollary \ref{cor1}, we get general
solutions for the Levi--Civita connection $\ ^{k}\mathbf{\nabla .}$
\end{lemma}

Finally, in this section we formulate:

\begin{conclusion}
\begin{enumerate}
\item Summarizing the Theorems \ref{th1} -- \ref{th3} and Lemma \ref{lm1},
we prove the Main Result stated in Theorem \ref{mth}.

\item The general solutions defined by the conditions of Theorem \ref{mth}
(and related results) can be extended to include contributions of an
arbitrary number of commutative and noncommutative parameteres. This is
possible following the constructions with Killing symmetries, in our case
for metrics (\ref{gendm}) provided in Ref. \cite{ijgmmp}, which can be
similarly reconsidered with higher order shells.
\end{enumerate}
\end{conclusion}

\section{Summary and Discussion}

\label{s5} This work was primarily motivated by the question if the Einstein
equations can be integrated in very general forms, for generic off--diagonal
metrics depending on all possible variables, in arbitrary dimensions. To the
best of our knowledge, such a problem has not yet been addressed in
mathematical and physical literature being known the high complexity of
related systems of nonlinear partial differential equations. This is in
spite of the fact that there were elaborated a number of analytic and
numerical methods of constructing exact and approximate solutions and that
various types of such solutions seem to be of crucial physical importance in
modern astrophysics and cosmology. Here we note that the bulk of former
derived solutions are for diagonalizabe metrics (by coordinate transforms),
depending on one and/or two (in some exceptional cases, on three) variables,
with compactified dimensions, imposed symmetries, boundary conditions etc.

In a series of works, see reviews of results in Refs. \cite{ijgmmp,vncg,vsgg}%
, one of the main our goals was to formulate a geometric method which would
allow us to construct exact solutions of gravitational field equations. We
applied the formalism of nonholonomic distributions with generating and
integration functions, when some of them are subjected to additional
conditions/ constraints (written as certain types of first order partial
equations, algebraic relations, symmetry conditions etc). That allowed us to
elaborate a general scheme for deriving exact solutions with one Killing
vector symmetry and various types of parametric dependencies. Finally, the
so--called anholonomic deformation method was developed for general
''non--Killing'' solutions in paper \cite{vgseg}.

Following the anholonomic deformation method, we define some 'more
convenient" holonomic and nonholonomic variables (frames coefficients and
coordinates), which for certain types well defined conditions transform the
Einstein equations into exactly integrable systems of equations. The key
idea is to use additionally some auxiliary linear connections
correspondingly adapted to nonholonmic distributions. Surprisingly, in our
approach, it was possible to reformulate (and, in general, to modify) the
Einstein equations in such forms, when general integral varieties can be
constructed. Subjecting the coefficients of such way defined solutions to
additional constraints, we can determine some integral subvarieties for
standard gravity theories and generalizations. Here we note that our
auxiliary connection (the so--called, canonical distinguished connection, in
brief, d--connection) is also metric compatible and uniquely defined by the
metric coefficients. It contains a nonholonomically induced torsion but such
a geometric object is completely different from that, for instance, in
Einstein--Cartan/ gauge / string theory. In our approach, we do not need any
additional field equations because we work with torsion coefficients induced
by certain off--diagonal coefficients of metric. All geometric constructions
can be equivalently performed using the Levi--Civita connection or,
alternatively, the canonical d--connection.

Of course, our findings should be considered only in a line of qualitative
understanding of the concept of general exact solutions in Einstein and high
dimensional gravity. For such generic nonlinear systems, it is not possible
to formulate any general uniqueness and completeness criteria for solution
if we do not introduce any additional suppositions on classes of generating
functions, symmetries, horizons, singularities, asymptotic conditions etc.
Only in some more special/ restricted cases, we can provide certain physical
meaning for such general classes of solutions; to put, for instance, the
Cauchy problem, construct some evolution models, determine symmetries of
interactions etc. Our constructions are general ones because "almost" any
solution in gravity theories can be parametrization in such a form at least
locally even very different classes of metrics and connections can be stated
globally for different topologies, boundary conditions, with various types
of horizons and singularities etc. It is not our aim to perform such studies
in this article.

Finally, we emphasize that the bulk of exact solutions in gravity theories
(in Einstein gravity and various supersymmetric/noncommutative sting, brane,
gauge, Kaluza--Klein, Lagrange--Finsler, generalizations etc) can be
represented in a form similar to (\ref{ansgensol}). In this paper, we do not
analyze possible explicit symmetries and physical properties of such
solutions. We consider such problems in our recent papers \cite%
{vrfs1,vrfs2,vrfs3,vrfs4,vrfs5,vkvnsst,vfbh,vfbhnc,avfbr} and plan to
provide further developments and applications in our future works.

\vskip5pt

\textbf{Acknowledgement: } Author is grateful to M. Anastasiei for
discussions and support.

\appendix

\setcounter{equation}{0} \renewcommand{\theequation}
{A.\arabic{equation}} \setcounter{subsection}{0}
\renewcommand{\thesubsection}
{A.\arabic{subsection}}

\section{Coefficients of N--adapted Curvature}

\label{asa}In this section, we outline some formulas which play an important
role in finding systems of partial differential equations which are
equivalent to the Einstein equations with nonholonomic variables of
arbitrary dimensions, see details in Refs. \cite{ijgmmp,vncg,vsgg,vrflg}.
The formulas for coefficients of curvature $\widehat{\mathcal{R}}$ of the
canonical d--connection $\widehat{\mathbf{D}}$ are written with respect to
N--adapted frames (\ref{dder}) and (\ref{ddif}).

\begin{theorem}
\label{thasa1}The curvature $\widehat{\mathcal{R}}$ (\ref{curv}) of $\ $the
canonical d--connection $\widehat{\mathbf{D}}$ computed with respect to
N--adapted frames (\ref{dder}) and (\ref{ddif}) is characterized by
coefficients
\begin{eqnarray}
\widehat{R}_{\ hjk}^{i} &=&e_{k}\widehat{L}_{\ hj}^{i}-e_{j}\widehat{L}_{\
hk}^{i}+\widehat{L}_{\ hj}^{m}\widehat{L}_{\ mk}^{i}-\widehat{L}_{\ hk}^{m}%
\widehat{L}_{\ mj}^{i}-\widehat{C}_{\ ha}^{i}\Omega _{\ kj}^{a},  \notag \\
\widehat{R}_{\ bjk}^{a} &=&e_{k}\widehat{L}_{\ bj}^{a}-e_{j}\widehat{L}_{\
bk}^{a}+\widehat{L}_{\ bj}^{c}\widehat{L}_{\ ck}^{a}-\widehat{L}_{\ bk}^{c}%
\widehat{L}_{\ cj}^{a}-\widehat{C}_{\ bc}^{a}\Omega _{\ kj}^{c},  \notag \\
\widehat{R}_{\ jka}^{i} &=&e_{a}\widehat{L}_{\ jk}^{i}-\widehat{D}_{k}%
\widehat{C}_{\ ja}^{i}+\widehat{C}_{\ jb}^{i}\widehat{T}_{\ ka}^{b},\
\label{dcurv} \\
\widehat{R}_{\ bka}^{c} &=&e_{a}\widehat{L}_{\ bk}^{c}-D_{k}\widehat{C}_{\
ba}^{c}+\widehat{C}_{\ bd}^{c}\widehat{T}_{\ ka}^{c},  \notag \\
\widehat{R}_{\ jbc}^{i} &=&e_{c}\widehat{C}_{\ jb}^{i}-e_{b}\widehat{C}_{\
jc}^{i}+\widehat{C}_{\ jb}^{h}\widehat{C}_{\ hc}^{i}-\widehat{C}_{\ jc}^{h}%
\widehat{C}_{\ hb}^{i},  \notag \\
\widehat{R}_{\ bcd}^{a} &=&e_{d}\widehat{C}_{\ bc}^{a}-e_{c}\widehat{C}_{\
bd}^{a}+\widehat{C}_{\ bc}^{e}\widehat{C}_{\ ed}^{a}-\widehat{C}_{\ bd}^{e}%
\widehat{C}_{\ ec}^{a};  \notag
\end{eqnarray}%
\begin{eqnarray*}
\widehat{R}_{\ \tau \beta \gamma }^{\alpha } &=&e_{\gamma }\widehat{L}_{\
\tau \beta }^{\alpha }-e_{\beta }\widehat{L}_{\ \tau \gamma }^{\alpha }+%
\widehat{L}_{\ \tau \beta }^{\mu }\widehat{L}_{\ \mu \gamma }^{\alpha }-%
\widehat{L}_{\ \tau \gamma }^{\mu }\widehat{L}_{\ \mu \beta }^{\alpha }-%
\widehat{C}_{\ \tau \ ^{1}a}^{\alpha }\Omega _{\ \gamma \beta }^{^{1}a}, \\
\widehat{R}_{\ ^{1}b\beta \gamma }^{^{1}a} &=&e_{\gamma }\widehat{L}_{\
^{1}b\beta }^{^{1}a}-e_{\beta }\widehat{L}_{\ ^{1}b\gamma }^{^{1}a}+\widehat{%
L}_{\ ^{1}b\beta }^{^{1}c}\widehat{L}_{\ ^{1}c\gamma }^{^{1}a}-\widehat{L}%
_{\ ^{1}b\gamma }^{^{1}c}\widehat{L}_{\ ^{1}c\beta }^{^{1}a}-\widehat{C}_{\
^{1}b\ ^{1}c}^{^{1}a}\Omega _{\ \gamma \beta }^{^{1}c}, \\
\widehat{R}_{\ \beta \gamma \ ^{1}a}^{\alpha } &=&e_{^{1}a}\widehat{L}_{\
\beta \gamma }^{\alpha }-\widehat{D}_{\gamma }\widehat{C}_{\ \beta \
^{1}a}^{\alpha }+\widehat{C}_{\ \beta \ ^{1}b}^{\alpha }\widehat{T}_{\
\gamma \ ^{1}a}^{\ ^{1}b},\  \\
\widehat{R}_{\ \ ^{1}b\gamma \ ^{1}a}^{\ ^{1}c} &=&e_{\ ^{1}a}\widehat{L}_{\
\ ^{1}b\gamma }^{\ ^{1}c}-\widehat{D}_{\gamma }\widehat{C}_{\ \ ^{1}b\
^{1}a}^{\ ^{1}c}+\widehat{C}_{\ ^{1}\ b\ ^{1}d}^{\ ^{1}c}\widehat{T}_{\
\gamma \ ^{1}a}^{\ ^{1}c},
\end{eqnarray*}%
\begin{eqnarray*}
\widehat{R}_{\ \beta \ ^{1}b\ ^{1}c}^{\alpha } &=&e_{\ ^{1}c}\widehat{C}_{\
\beta \ ^{1}b}^{\alpha }-e_{\ ^{1}b}\widehat{C}_{\ \beta \ ^{1}c}^{\alpha }+%
\widehat{C}_{\ \beta \ ^{1}b}^{\mu }\widehat{C}_{\ \mu \ ^{1}c}^{\alpha }\ -%
\widehat{C}_{\ \beta \ ^{1}c}^{\mu }\widehat{C}_{\ \mu \ ^{1}b}^{\alpha }, \\
\widehat{R}_{\ \ ^{1}b\ ^{1}c\ ^{1}d}^{\ ^{1}a} &=&e_{\ ^{1}d}\widehat{C}_{\
\ ^{1}b\ ^{1}c}^{\ ^{1}a}-e_{\ ^{1}c}\widehat{C}_{\ \ ^{1}b\ ^{1}d}^{\ ^{1}a}
\\
&&+\widehat{C}_{\ \ ^{1}b\ ^{1}c}^{\ ^{1}e}\widehat{C}_{\ \ ^{1}e\ ^{1}d}^{\
^{1}a}-\widehat{C}_{\ \ ^{1}b\ ^{1}d}^{\ ^{1}e}\widehat{C}_{\ \ ^{1}e\
^{1}c}^{\ ^{1}a};
\end{eqnarray*}%
\begin{eqnarray*}
\widehat{R}_{\ \ ^{1}\tau \ ^{1}\beta \ ^{1}\gamma }^{\ ^{1}\alpha } &=&e_{\
^{1}\gamma }\widehat{L}_{\ \ ^{1}\tau \ ^{1}\beta }^{\ ^{1}\alpha }-e_{\
^{1}\beta }\widehat{L}_{\ \ ^{1}\tau \ ^{1}\gamma }^{\ ^{1}\alpha }+\widehat{%
L}_{\ \ ^{1}\tau \ ^{1}\beta }^{\ ^{1}\mu }\widehat{L}_{\ \ ^{1}\mu \
^{1}\gamma }^{\ ^{1}\alpha } \\
&&-\widehat{L}_{\ \ ^{1}\tau \ ^{1}\gamma }^{\ ^{1}\mu }\widehat{L}_{\ \
^{1}\mu \ ^{1}\beta }^{\ ^{1}\alpha }-\widehat{C}_{\ \ ^{1}\tau \ ^{2}a}^{\
^{1}\alpha }\Omega _{\ \ ^{1}\gamma \ ^{1}\beta }^{^{2}a}, \\
\widehat{R}_{\ ^{2}b\ ^{1}\beta \ ^{1}\gamma }^{^{2}a} &=&e_{\ ^{1}\gamma }%
\widehat{L}_{\ ^{2}b\ ^{1}\beta }^{^{2}a}-e_{\ ^{1}\beta }\widehat{L}_{\
^{2}b\ ^{1}\gamma }^{^{2}a} \\
&&+\widehat{L}_{\ ^{2}b\ ^{1}\beta }^{^{2}c}\widehat{L}_{\ ^{2}c\ ^{1}\gamma
}^{^{2}a}-\widehat{L}_{\ ^{2}b\ ^{1}\gamma }^{^{2}c}\widehat{L}_{\ ^{2}c\
^{1}\beta }^{^{2}a}-\widehat{C}_{\ ^{2}b\ ^{2}c}^{^{2}a}\Omega _{\ \
^{1}\gamma \ ^{1}\beta }^{^{2}c},
\end{eqnarray*}%
\begin{eqnarray*}
\widehat{R}_{\ \ ^{1}\beta \ ^{1}\gamma \ ^{2}a}^{\ ^{1}\alpha } &=&e_{^{2}a}%
\widehat{L}_{\ \ ^{1}\beta \ ^{1}\gamma }^{\ ^{1}\alpha }-\widehat{D}_{\
^{1}\gamma }\widehat{C}_{\ \ ^{1}\beta \ ^{2}a}^{\ ^{1}\alpha }+\widehat{C}%
_{\ \ ^{1}\beta \ ^{2}b}^{\ ^{1}\alpha }\widehat{T}_{\ \ ^{1}\gamma \
^{2}a}^{\ ^{2}b},\  \\
\widehat{R}_{\ \ ^{2}b\ ^{1}\gamma \ ^{2}a}^{\ ^{2}c} &=&e_{\ ^{2}a}\widehat{%
L}_{\ \ ^{2}b\ ^{1}\gamma }^{\ ^{2}c}-\widehat{D}_{\ ^{1}\gamma }\widehat{C}%
_{\ \ ^{2}b\ ^{2}a}^{\ ^{2}c}+\widehat{C}_{\ ^{2}\ b\ ^{2}d}^{\ ^{2}c}%
\widehat{T}_{\ \ ^{1}\gamma \ ^{2}a}^{\ ^{2}c},
\end{eqnarray*}%
\begin{eqnarray*}
\widehat{R}_{\ \ ^{1}\beta \ ^{2}b\ ^{2}c}^{\ ^{1}\alpha } &=&e_{\ ^{2}c}%
\widehat{C}_{\ \ ^{1}\beta \ ^{2}b}^{\alpha }-e_{\ ^{2}b}\widehat{C}_{\ \
^{1}\beta \ ^{2}c}^{\ ^{1}\alpha } \\
&&+\widehat{C}_{\ \ ^{1}\beta \ ^{2}b}^{\ ^{1}\mu }\widehat{C}_{\ \ ^{1}\mu
\ ^{2}c}^{\ ^{1}\alpha }\ -\widehat{C}_{\ \ ^{1}\beta \ ^{2}c}^{\ ^{1}\mu }%
\widehat{C}_{\ \ ^{1}\mu \ ^{2}b}^{\ ^{1}\alpha }, \\
\widehat{R}_{\ \ ^{2}b\ ^{2}c\ ^{2}d}^{\ ^{2}a} &=&e_{\ ^{2}d}\widehat{C}_{\
\ ^{2}b\ ^{2}c}^{\ ^{2}a}-e_{\ ^{2}c}\widehat{C}_{\ \ ^{2}b\ ^{2}d}^{\ ^{2}a}
\\
&&+\widehat{C}_{\ \ ^{2}b\ ^{2}c}^{\ ^{2}e}\widehat{C}_{\ \ ^{2}e\ ^{2}d}^{\
^{2}a}-\widehat{C}_{\ \ ^{2}b\ ^{2}d}^{\ ^{2}e}\widehat{C}_{\ \ ^{2}e\
^{2}c}^{\ ^{2}a};\\
&& .......... \\
\widehat{R}_{\ \ ^{k-1}\tau \ ^{k-1}\beta \ ^{k-1}\gamma }^{\ ^{k-1}\alpha }
&=&e_{\ ^{k-1}\gamma }\widehat{L}_{\ \ ^{k-1}\tau \ ^{k-1}\beta }^{\
^{k-1}\alpha }-e_{\ ^{k-1}\beta }\widehat{L}_{\ \ ^{k-1}\tau \ ^{k-1}\gamma
}^{\ ^{k-1}\alpha }+ \\
&&\widehat{L}_{\ \ ^{k-1}\tau \ ^{k-1}\beta }^{\ ^{k-1}\mu }\widehat{L}_{\
\ ^{k-1}\mu \ ^{k-1}\gamma }^{\ ^{k-1}\alpha }-\widehat{L}_{\ \ ^{k-1}\tau \
^{k-1}\gamma }^{\ ^{k-1}\mu }\widehat{L}_{\ \ ^{k-1}\mu \ ^{k-1}\beta }^{\
^{k-1}\alpha } \\
&&-\widehat{C}_{\ \ ^{k-1}\tau \ ^{k}a}^{\ ^{k-1}\alpha }\Omega _{\ \
^{k-1}\gamma \ ^{k-1}\beta }^{^{k}a}, \\
\widehat{R}_{\ ^{k}b\ ^{k-1}\beta \ ^{k-1}\gamma }^{^{k}a} &=&e_{\
^{k-1}\gamma }\widehat{L}_{\ ^{k-1}b\ ^{k-1}\beta }^{^{k-k}a}-e_{\
^{k-1}\beta }\widehat{L}_{\ ^{k}b\ ^{k-1}\gamma }^{^{k}a} \\
&&+\widehat{L}_{\ ^{k}b\ ^{k-1}\beta }^{^{k}c}\widehat{L}_{\ ^{k}c\
^{k-1}\gamma }^{^{k}a}-\widehat{L}_{\ ^{k}b\ ^{k-1}\gamma }^{^{k}c}\widehat{L%
}_{\ ^{k}c\ ^{k-1}\beta }^{^{k}a} \\
&&-\widehat{C}_{\ ^{k}b\ ^{k}c}^{^{k}a}\Omega _{\ \ ^{k-1}\gamma \
^{k-1}\beta }^{^{k}c},\\
\widehat{R}_{\ \ ^{k-1}\beta \ ^{k-1}\gamma \ ^{k}a}^{\ ^{k-1}\alpha }
&=&e_{^{k}a}\widehat{L}_{\ \ ^{k-1}\beta \ ^{k-1}\gamma }^{\ ^{k-1}\alpha }-%
\widehat{D}_{\ ^{k-1}\gamma }\widehat{C}_{\ \ ^{k-1}\beta \ ^{k}a}^{\
^{k-1}\alpha } \\
&&+\widehat{C}_{\ \ ^{k-1}\beta \ ^{k}b}^{\ ^{k-1}\alpha }\widehat{T}_{\ \
^{k-1}\gamma \ ^{k}a}^{\ ^{k}b}, \\
\widehat{R}_{\ \ ^{k}b\ ^{k-1}\gamma \ ^{k}a}^{\ ^{k}c} &=&e_{\ ^{k}a}%
\widehat{L}_{\ \ ^{k}b\ ^{k-1}\gamma }^{\ ^{k}c}-\widehat{D}_{\ ^{k-1}\gamma
}\widehat{C}_{\ \ ^{k}b\ ^{k}a}^{\ ^{k}c}+\widehat{C}_{\ ^{k}\ b\ ^{k}d}^{\
^{k}c}\widehat{T}_{\ \ ^{k-1}\gamma \ ^{k}a}^{\ ^{k}c},\\
\widehat{R}_{\ \ ^{k-1}\beta \ ^{k}b\ ^{k}c}^{\ ^{k-1}\alpha } &=&e_{\ ^{k}c}%
\widehat{C}_{\ \ ^{k-1}\beta \ ^{k}b}^{\alpha }-e_{\ ^{k}b}\widehat{C}_{\ \
^{k-1}\beta \ ^{k}c}^{\ ^{k-1}\alpha } \\
&&+\widehat{C}_{\ \ ^{k-1}\beta \ ^{k}b}^{\ ^{k-1}\mu }\widehat{C}_{\ \
^{k-1}\mu \ ^{k}c}^{\ ^{k-1}\alpha }\ -\widehat{C}_{\ \ ^{k-1}\beta \
^{k}c}^{\ ^{k-1}\mu }\widehat{C}_{\ \ ^{k-1}\mu \ ^{k}b}^{\ ^{k-1}\alpha },
\\
\widehat{R}_{\ \ ^{k}b\ ^{k}c\ ^{k}d}^{\ ^{k}a} &=&e_{\ ^{k}d}\widehat{C}_{\
\ ^{k}b\ ^{k}c}^{\ ^{k}a}-e_{\ ^{k}c}\widehat{C}_{\ \ ^{k}b\ ^{k}d}^{\ ^{k}a}
\\
&&+\widehat{C}_{\ \ ^{k}b\ ^{k}c}^{\ ^{k}e}\widehat{C}_{\ \ ^{k}e\ ^{k}d}^{\
^{k}a}-\widehat{C}_{\ \ ^{k}b\ ^{k}d}^{\ ^{k}e}\widehat{C}_{\ \ ^{k}e\
^{k}c}^{\ ^{k}a}.
\end{eqnarray*}
\end{theorem}

\begin{proof}
It follows from ''shell by shell computations'' as in Refs. \cite%
{ijgmmp,vncg,vsgg,vrflg,vmth1,vmth2,vmth3,vnp1,vsp1,vv1,vst,vvc,vd1,vd2}.$\
\square $
\end{proof}

\setcounter{equation}{0} \renewcommand{\theequation}
{B.\arabic{equation}} \setcounter{subsection}{0}
\renewcommand{\thesubsection}
{B.\arabic{subsection}}

\section{Proof of Theorem \ref{th2}}

\label{proofth2}Such a proof can be obtained by straightforward computations
as in Parts I and II of monograph \cite{vsgg}, containing all developments
from Refs. \cite{vmth1,vmth2,vmth3}, see also summaries and some important
details and discussions in Refs. \cite{ijgmmp,vncg}. In this section, we
generalize some formulas by considering ''shell'' labels for indices, when $%
k=0,1,2,...$ using data (\ref{data1}) for a metric $\ ^{k}\mathbf{g=\mathbf{%
\{\mathbf{g}}}_{\ ^{k}\beta \ ^{k}\gamma }\mathbf{\mathbf{\}}}$ (\ref{gendm}%
).

We can perform a N--adapted differential calculus on a N--anholonomic
manifold if instead of partial derivatives $\partial _{\ ^{k}\alpha
}=\partial /\partial u^{\ ^{k}\alpha }$ there are considered operators (\ref%
{dder}) parametrized in the form
 $\mathbf{e}\ _{^{k-1}\alpha }=\partial _{\ ^{k-1}\alpha }-N_{\ ^{k-1}\alpha
}^{\ ^{k}a}\partial _{\ ^{k}a}=\partial _{\ ^{k-1}\alpha }-w_{\ ^{k-1}\alpha
}^{\ }\partial _{\ ^{k}v}-n_{\ ^{k-1}\alpha }^{\ }\partial _{\ ^{k}y},$
 for $y^{4+2k}=\ ^{k}v$ and $y^{5+2k}=\ ^{k}y.$ For instance, for data (\ref%
{data1}), the coefficients of N--connection curvature (\ref{ncurv}) are
\begin{eqnarray}
\Omega _{\ ^{k-1}\alpha \ ^{k-1}\beta }^{\ 4+2k} &=&\partial _{\
^{k-1}\alpha }w_{\ ^{k-1}\beta }^{\ }-\partial _{\ ^{k-1}\beta }w_{\
^{k-1}\alpha }^{\ }  \notag \\
&& -w_{\ ^{k-1}\alpha }^{\ }\partial _{\ ^{k}v}w_{\ ^{k-1}\beta }^{\ }+w_{\
^{k-1}\beta }^{\ }\partial _{\ ^{k}v}w_{\ ^{k-1}\alpha }^{\ };  \label{omeg}
\\
\Omega _{\ ^{k-1}\alpha \ ^{k-1}\beta }^{\ 5+2k} &=&\partial _{\
^{k-1}\alpha }n_{\ ^{k-1}\beta }^{\ }-\partial _{\ ^{k-1}\beta }n_{\
^{k-1}\alpha }^{\ }  \notag \\
&& -w_{\ ^{k-1}\alpha }^{\ }\partial _{\ ^{k}v}n_{\ ^{k-1}\beta }^{\ }+w_{\
^{k-1}\beta }^{\ }\partial _{\ ^{k}v}n_{\ ^{k-1}\alpha }^{\ }.  \notag
\end{eqnarray}

In a similar form we compute all coefficients of the canonical d--connection
(\ref{candcon}) and its Ricci and Einstein tensors.

\subsubsection*{Coefficients of the canonical d--connection}

For data (\ref{data1}), we get such nontrivial coefficients of $\widehat{%
\mathbf{\Gamma }}_{\ \ ^{k}\alpha \ ^{k}\beta }^{\ ^{k}\gamma }:$
\begin{eqnarray}
&&\widehat{L}_{22}^{2} =\frac{\partial _{2}g_{2}}{2g_{2}},\ \widehat{L}%
_{23}^{2}=\frac{\partial _{3}g_{2}}{2g_{2}},\ \widehat{L}_{33}^{2}=-\frac{%
\partial _{2}g_{3}}{2g_{2}},\ \widehat{L}_{22}^{3} =-\frac{\partial _{3}g_{2}%
}{2g_{3}},\   \label{aux1a} \\
&&\widehat{L} _{23}^{3}=\frac{\partial _{2}g_{3}}{2g_{3}},\ \widehat{L}%
_{33}^{3}=\frac{\partial _{3}g_{3}}{2g_{3}},\ \widehat{L}_{4i}^{4} =\frac{1}{%
2h_{4}}\left( \partial _{i}h_{4}-w_{i}\partial _{v}h_{4}\right);\
 \widehat{L}_{4j}^{5}=\frac{1}{2}\partial _{v}n_{j},\  \notag \\
&& \widehat{L}%
_{5j}^{5} = \frac{1}{2h_{5}}\left( \partial _{j}h_{5}-w_{j}\partial
_{v}h_{5}\right);\
\widehat{C}_{44}^{4} =\frac{\partial _{v}h_{4}}{2h_{4}},\widehat{C}%
_{55}^{4}=-\frac{\partial _{v}h_{5}}{2h_{4}},\widehat{C}_{45}^{5}=\frac{%
\partial _{v}h_{5}}{2h_{5}};  \notag\\
&&.... \notag \\
&&\widehat{L}_{4+2k\ ^{k-1}\alpha }^{4+2k} =\frac{1}{2h_{4+2k}}\left(
\partial _{\ ^{k-1}\alpha }h_{4+2k}-w_{\ ^{k-1}\alpha }\partial _{\
^{k}v}h_{4+2k}\right),\ \widehat{L}_{4+2k\ ^{k-1}\alpha }^{5+2k} = \notag \\ && \frac{1}{2}\partial _{\
^{k}v}n_{\ ^{k-1}\alpha },\
\ \widehat{L}_{5+2k\ ^{k-1}\alpha }^{5+2k} = \frac{1}{2h_{5+2k}}\left(
\partial _{\ ^{k-1}\alpha }h_{5+2k}-w_{\ ^{k-1}\alpha }\partial _{\
^{k}v}h_{5+2k}\right); \notag \\
&&\widehat{C}_{4+2k\ 4+2k}^{4+2k} =\frac{\partial _{\ ^{k}v}h_{4+2k}}{%
2h_{4+2k}},\widehat{C}_{5+2k\ 5+2k}^{4+2k}=-\frac{\partial _{\ ^{k}v}h_{5+2k}%
}{2h_{4+2k}}, \notag \\
&&\widehat{C}_{4+2k\ 5+2k}^{5+2k} =\frac{\partial _{\ ^{k}v}h_{5+2k}}{%
2h_{5+2k}}. \notag
\end{eqnarray}%
We note that
\begin{equation}
\widehat{C}_{jc}^{i}=\frac{1}{2}g^{ik}\frac{\partial g_{jk}}{\partial y^{c}}%
=0,...,\widehat{C}_{\ ^{k-1}\beta \ ^{k}c}^{\ ^{k-1}\alpha }=\frac{1}{2}g^{\
^{k-1}\alpha \ ^{k-1}\tau \ }\ \frac{\partial g_{^{k-1}\beta \ ^{k-1}\tau }}{%
\partial y^{\ ^{k}c}}=0,  \label{aux02cc}
\end{equation}%
which is an important condition for generating exact solutions of the
Einstein equations for the Levi--Civita connection, see formulas (\ref%
{lccondg}).

\subsubsection*{Calculation of torsion coefficients}

The nontrivial coefficients of torsions (\ref{dtor}) for data (\ref{data1})
are given by formulas (\ref{omeg}) and, respectively, (\ref{aux1a})
resulting in
\begin{eqnarray}
\widehat{T}_{\ ^{k-1}\alpha \ ^{k-1}\beta }^{\ 4+2k} &=&\partial _{\
^{k-1}\beta }w_{\ ^{k-1}\alpha }^{\ }-\partial _{\ ^{k-1}\alpha }w_{\
^{k-1}\beta }^{\ }  \label{dtorsc} \\
&&-w_{\ ^{k-1}\beta }^{\ }\partial _{\ ^{k}v}w_{\ ^{k-1}\alpha }^{\ }+w_{\
^{k-1}\alpha }^{\ }\partial _{\ ^{k}v}w_{\ ^{k-1}\beta }^{\ };  \notag \\
\widehat{T}_{\ ^{k-1}\alpha \ ^{k-1}\beta }^{\ 5+2k} &=&\partial _{\
^{k-1}\beta }n_{\ ^{k-1}\alpha }^{\ }-\partial _{\ ^{k-1}\alpha }n_{\
^{k-1}\beta }^{\ }  \notag \\
&&w_{\ ^{k-1}\beta }^{\ }\partial _{\ ^{k}v}n_{\ ^{k-1}\alpha }^{\ }-w_{\
^{k-1}\alpha }^{\ }\partial _{\ ^{k}v}n_{\ ^{k-1}\beta }^{\ },  \notag \\
\widehat{T}_{4+2k\ ^{k-1}\alpha }^{4+2k} &=&\partial _{\ ^{k}v}w_{\
^{k-1}\alpha }-\frac{1}{2h_{4+2k}}\left( \partial _{\ ^{k-1}\alpha
}h_{4+2k}-w_{\ ^{k-1}\alpha }\partial _{\ ^{k}v}h_{4+2k}\right) ,  \notag \\
\widehat{T}_{5+2k\ ^{k-1}\alpha }^{4+2k} &=&\frac{h_{5+2k}}{2h_{4+2k}}%
\partial _{\ ^{k}v}n_{\ ^{k-1}\alpha },\ \widehat{T}_{4+2k\ ^{k-1}\alpha
}^{5+2k}=\frac{1}{2}\partial _{\ ^{k}v}n_{\ ^{k-1}\alpha },  \notag \\
\widehat{T}_{5+2k\ ^{k-1}\alpha }^{5+2k} &=&-\frac{1}{2h_{5+2k}}\left(
\partial _{\ ^{k-1}\alpha }h_{5+2k}-w_{\ ^{k-1}\alpha }\partial _{\
^{k}v}h_{5+2k}\right) .  \notag
\end{eqnarray}

\subsubsection*{Calculation of the Ricci tensor}

For instance, let us compute the values $\widehat{R}_{ij}=\widehat{R}_{\
ijk}^{k}$ from (\ref{dricci}),
\begin{equation*}
\widehat{R}_{\ hjk}^{i}=\mathbf{e}_{k}\widehat{L}_{.hj}^{i}-\mathbf{e}_{j}%
\widehat{L}_{.hk}^{i}+\widehat{L}_{.hj}^{m}\widehat{L}_{mk}^{i}-\widehat{L}%
_{.hk}^{m}\widehat{L}_{mj}^{i}-\widehat{C}_{.ha}^{i}\Omega _{.jk}^{a},
\end{equation*}%
using (\ref{dcurv}) and $\widehat{C}_{.ha}^{i}=0$ (\ref{aux02cc}). We have
 $\mathbf{e}_{k}\widehat{L}_{.hj}^{i}=\partial _{k}\widehat{L}%
_{.hj}^{i}+N_{k}^{a}\partial _{a}\widehat{L}_{.hj}^{i}=\partial _{k}\widehat{%
L}_{.hj}^{i}+w_{k}\left( \widehat{L}_{.hj}^{i}\right) ^{\ast }=\partial _{k}%
\widehat{L}_{.hj}^{i}$ %
because $\widehat{L}_{.hj}^{i}$ do not depend on variable $y^{4}=v.$ We use,
in brief, denotations of type $\partial _{2}g=g^{\bullet },\partial
_{3}g=g^{^{\prime }},\partial _{4}g=g^{\ast }.$

Deriving (\ref{aux1a}), we obtain
\begin{eqnarray*}
\partial _{2}\widehat{L}_{\ 22}^{2} &=&\frac{g_{2}^{\bullet \bullet }}{2g_{2}%
}-\frac{\left( g_{2}^{\bullet }\right) ^{2}}{2\left( g_{2}\right) ^{2}},\
\partial _{2}\widehat{L}_{\ 23}^{2}=\frac{g_{2}^{\bullet ^{\prime }}}{2g_{2}}%
-\frac{g_{2}^{\bullet }g_{2}^{^{\prime }}}{2\left( g_{2}\right) ^{2}},\  \\
\partial _{2}\widehat{L}_{\ 33}^{2} &=&-\frac{g_{3}^{\bullet \bullet }}{%
2g_{2}}+\frac{g_{2}^{\bullet }g_{3}^{\bullet }}{2\left( g_{2}\right) ^{2}},\
\partial _{2}\widehat{L}_{\ 22}^{3}=-\frac{g_{2}^{\bullet ^{\prime }}}{2g_{3}%
}+\frac{g_{2}^{\bullet }g_{3}^{^{\prime }}}{2\left( g_{3}\right) ^{2}}, \\
\partial _{2}\widehat{L}_{\ 23}^{3} &=&\frac{g_{3}^{\bullet \bullet }}{2g_{3}%
}-\frac{\left( g_{3}^{\bullet }\right) ^{2}}{2\left( g_{3}\right) ^{2}},\
\partial _{2}\widehat{L}_{\ 33}^{3}=\frac{g_{3}^{\bullet ^{\prime }}}{2g_{3}}%
-\frac{g_{3}^{\bullet }g_{3}^{^{\prime }}}{2\left( g_{3}\right) ^{2}},
\end{eqnarray*}%
\begin{eqnarray*}
\partial _{3}\widehat{L}_{\ 22}^{2} &=&\frac{g_{2}^{\bullet ^{\prime }}}{%
2g_{2}}-\frac{g_{2}^{\bullet }g_{2}^{^{\prime }}}{2\left( g_{2}\right) ^{2}}%
,\partial _{3}\widehat{L}_{\ 23}^{2}=\frac{g_{2}^{ll}}{2g_{2}}-\frac{\left(
g_{2}^{l}\right) ^{2}}{2\left( g_{2}\right) ^{2}}, \\
\partial _{3}\widehat{L}_{\ 33}^{2} &=&-\frac{g_{3}^{\bullet ^{\prime }}}{%
2g_{2}}+\frac{g_{3}^{\bullet }g_{2}^{^{\prime }}}{2\left( g_{2}\right) ^{2}}%
,\ \partial _{3}\widehat{L}_{\ 22}^{3}=-\frac{g_{2}^{^{\prime \prime }}}{%
2g_{3}}+\frac{g_{2}^{\bullet }g_{2}^{^{\prime }}}{2\left( g_{3}\right) ^{2}},
\\
\partial _{3}\widehat{L}_{\ 23}^{3} &=&\frac{g_{3}^{\bullet ^{\prime }}}{%
2g_{3}}-\frac{g_{3}^{\bullet }g_{3}^{^{\prime }}}{2\left( g_{3}\right) ^{2}}%
,\partial _{3}\widehat{L}_{\ 33}^{3}=\frac{g_{3}^{ll}}{2g_{3}}-\frac{\left(
g_{3}^{l}\right) ^{2}}{2\left( g_{3}\right) ^{2}}.
\end{eqnarray*}

For these values, there are only 2 nontrivial components,
\begin{eqnarray*}
\widehat{R}_{\ 323}^{2} &=&\frac{g_{3}^{\bullet \bullet }}{2g_{2}}-\frac{%
g_{2}^{\bullet }g_{3}^{\bullet }}{4\left( g_{2}\right) ^{2}}-\frac{\left(
g_{3}^{\bullet }\right) ^{2}}{4g_{2}g_{3}}+\frac{g_{2}^{ll}}{2g_{2}}-\frac{%
g_{2}^{l}g_{3}^{l}}{4g_{2}g_{3}}-\frac{\left( g_{2}^{l}\right) ^{2}}{4\left(
g_{2}\right) ^{2}} \\
\widehat{R}_{\ 223}^{3} &=&-\frac{g_{3}^{\bullet \bullet }}{2g_{3}}+\frac{%
g_{2}^{\bullet }g_{3}^{\bullet }}{4g_{2}g_{3}}+\frac{\left( g_{3}^{\bullet
}\right) ^{2}}{4(g_{3})^{2}}-\frac{g_{2}^{ll}}{2g_{3}}+\frac{%
g_{2}^{l}g_{3}^{l}}{4(g_{3})^{2}}+\frac{\left( g_{2}^{l}\right) ^{2}}{%
4g_{2}g_{3}}
\end{eqnarray*}%
with $\ \widehat{R}_{22}=-\widehat{R}_{\ 223}^{3}$ and $\widehat{R}_{33}=%
\widehat{R}_{\ 323}^{2},$ or%
\begin{equation*}
\widehat{R}_{2}^{2}=\widehat{R}_{3}^{3}=-\frac{1}{2g_{2}g_{3}}\left[
g_{3}^{\bullet \bullet }-\frac{g_{2}^{\bullet }g_{3}^{\bullet }}{2g_{2}}-%
\frac{\left( g_{3}^{\bullet }\right) ^{2}}{2g_{3}}+g_{2}^{\prime \prime }-%
\frac{g_{2}^{l}g_{3}^{l}}{2g_{3}}-\frac{\left( g_{2}^{l}\right) ^{2}}{2g_{2}}%
\right]
\end{equation*}%
as in (\ref{4ep1a}).

Now, we consider
\begin{eqnarray*}
\widehat{R}_{\ bka}^{c} &=&\frac{\partial \widehat{L}_{.bk}^{c}}{\partial
y^{a}}-\left( \frac{\partial \widehat{C}_{.ba}^{c}}{\partial x^{k}}+\widehat{%
L}_{.dk}^{c\,}\widehat{C}_{.ba}^{d}-\widehat{L}_{.bk}^{d}\widehat{C}%
_{.da}^{c}-\widehat{L}_{.ak}^{d}\widehat{C}_{.bd}^{c}\right) +\widehat{C}%
_{.bd}^{c}\widehat{T}_{.ka}^{d} \\
&=&\frac{\partial \widehat{L}_{.bk}^{c}}{\partial y^{a}}-\widehat{C}%
_{.ba|k}^{c}+\widehat{C}_{.bd}^{c}\widehat{T}_{.ka}^{d}
\end{eqnarray*}%
from (\ref{dcurv}). Contracting indices, we get
 $\widehat{R}_{bk}=\widehat{R}_{\ bka}^{a}=\frac{\partial L_{.bk}^{a}}{%
\partial y^{a}}-\widehat{C}_{.ba|k}^{a}+\widehat{C}_{.bd}^{a}\widehat{T}%
_{.ka}^{d}.$
Let us denote $\widehat{C}_{b}=\widehat{C}_{.ba}^{c}$ and write%
 $\widehat{C}_{.b|k}=\mathbf{e}_{k}\widehat{C}_{b}-\widehat{L}_{\ bk}^{d\,}%
\widehat{C}_{d}=\partial _{k}\widehat{C}_{b}-N_{k}^{e}\partial _{e}\widehat{C%
}_{b}-\widehat{L}_{\ bk}^{d\,}\widehat{C}_{d}=\partial _{k}\widehat{C}%
_{b}-w_{k}\widehat{C}_{b}^{\ast }-\widehat{L}_{\ bk}^{d\,}\widehat{C}_{d}.$
 We express $ \widehat{R}_{bk}=\ _{[1]}R_{bk}+\ _{[2]}R_{bk}+\ _{[3]}R_{bk},$
where%
\begin{eqnarray*}
\ _{[1]}R_{bk} &=&\left( \widehat{L}_{bk}^{4}\right) ^{\ast },\
_{[2]}R_{bk}=-\partial _{k}\widehat{C}_{b}+w_{k}\widehat{C}_{b}^{\ast }+%
\widehat{L}_{\ bk}^{d\,}\widehat{C}_{d}, \\
\ _{[3]}R_{bk} &=&\widehat{C}_{.bd}^{a}\widehat{T}_{.ka}^{d}=\widehat{C}%
_{.b4}^{4}\widehat{T}_{.k4}^{4}+\widehat{C}_{.b5}^{4}\widehat{T}_{.k4}^{5}+%
\widehat{C}_{.b4}^{5}\widehat{T}_{.k5}^{4}+\widehat{C}_{.b5}^{5}\widehat{T}%
_{.k5}^{5}
\end{eqnarray*}%
for $\ \widehat{C}_{4}=\widehat{C}_{44}^{4}+\widehat{C}_{45}^{5}=\frac{%
h_{4}^{\ast }}{2h_{4}}+\frac{h_{5}^{\ast }}{2h_{5}},\widehat{C}_{5}=\widehat{%
C}_{54}^{4}+\widehat{C}_{55}^{5}=0.$

We compute $\ \widehat{R}_{4k}=\ _{[1]}R_{4k}+\ _{[2]}R_{4k}+\ _{[3]}R_{4k}$
with
\begin{eqnarray*}
\ _{[1]}R_{4k} &=&\left( \widehat{L}_{4k}^{4}\right) ^{\ast },\
_{[2]}R_{4k}=-\partial _{k}\widehat{C}_{4}+w_{k}\widehat{C}_{4}^{\ast }+%
\widehat{L}_{\ 4k}^{4\,}\widehat{C}_{4}, \\
\ _{[3]}R_{4k} &=&\widehat{C}_{.44}^{4}\widehat{T}_{.k4}^{4}+\widehat{C}%
_{.45}^{4}\widehat{T}_{.k4}^{5}+\widehat{C}_{.44}^{5}\widehat{T}_{.k5}^{4}+%
\widehat{C}_{.45}^{5}\widehat{T}_{.k5}^{5}.
\end{eqnarray*}%
Summarizing, we get%
\begin{equation*}
2h_{5}\widehat{R}_{4k}=w_{k}\left[ h_{5}^{\ast \ast }-\frac{\left(
h_{5}^{\ast }\right) ^{2}}{2h_{5}}-\frac{h_{4}^{\ast }h_{5}^{\ast }}{2h_{4}}%
\right] +\frac{h_{5}^{\ast }}{2}\left( \frac{\partial _{k}h_{4}}{h_{4}}+%
\frac{\partial _{k}h_{5}}{h_{5}}\right) -\partial _{k}h_{5}^{\ast }
\end{equation*}%
which is equivalent to (\ref{4ep3a}).

In a similar way, we compute $\ \widehat{R}_{5k}=\ _{[1]}R_{5k}+\
_{[2]}R_{5k}+\ _{[3]}R_{5k},$ where
\begin{eqnarray*}
\ _{[1]}R_{5k} &=&\left( \widehat{L}_{5k}^{4}\right) ^{\ast },\
_{[2]}R_{5k}=-\partial _{k}\widehat{C}_{5}+w_{k}\widehat{C}_{5}^{\ast }+%
\widehat{L}_{\ 5k}^{4\,}\widehat{C}_{4}, \\
\ _{[3]}R_{5k} &=&\widehat{C}_{.54}^{4}\widehat{T}_{.k4}^{4}+\widehat{C}%
_{.55}^{4}\widehat{T}_{.k4}^{5}+\widehat{C}_{.54}^{5}\widehat{T}_{.k5}^{4}+%
\widehat{C}_{.55}^{5}\widehat{T}_{.k5}^{5}.
\end{eqnarray*}%
We have
 $\widehat{R}_{5k} = \left( \widehat{L}_{5k}^{4}\right) ^{\ast }+\widehat{L}%
_{\ 5k}^{4\,}\widehat{C}_{4}+\widehat{C}_{.55}^{4}\widehat{T}_{.k4}^{5}+%
\widehat{C}_{.54}^{5}\widehat{T}_{.k5}^{4}=
\left( -\frac{h_{5}}{h_{4}}n_{k}^{\ast }\right) ^{\ast }-\frac{h_{5}}{%
h_{4}}n_{k}^{\ast }\left( \frac{h_{4}^{\ast }}{2h_{4}}+\frac{h_{5}^{\ast }}{%
2h_{5}}\right) +\frac{h_{5}^{\ast }}{2h_{5}}\frac{h_{5}}{2h_{4}}n_{k}^{\ast
}-\frac{h_{5}^{\ast }}{2h_{4}}\frac{1}{2}n_{k}^{\ast },$
which can be written%
\begin{equation*}
2h_{4}\widehat{R}_{5k}=h_{5}n_{k}^{\ast \ast }+\left( \frac{h_{5}}{h_{4}}%
h_{4}^{\ast }-\frac{3}{2}h_{5}^{\ast }\right) n_{k}^{\ast },
\end{equation*}%
i. e. we prove (\ref{4ep4a}).

For $
\widehat{R}_{\ jka}^{i}=\frac{\partial \widehat{L}_{.jk}^{i}}{\partial y^{k}}%
-\left( \frac{\partial \widehat{C}_{.ja}^{i}}{\partial x^{k}}+\widehat{L}%
_{.lk}^{i}\widehat{C}_{.ja}^{l}-\widehat{L}_{.jk}^{l}\widehat{C}_{.la}^{i}-%
\widehat{L}_{.ak}^{c}\widehat{C}_{.jc}^{i}\right) +\widehat{C}_{.jb}^{i}%
\widehat{T}_{.ka}^{b}$
from (\ref{dcurv}), we obtain zeros because $\widehat{C}_{.jb}^{i}=0$ and $%
\widehat{L}_{.jk}^{i}$ do not depend on $y^{k}.$ So,
 $\widehat{R}_{ja}=\widehat{R}_{\ jia}^{i}=0.$

Taking
$\widehat{R}_{\ bcd}^{a}=\frac{\partial \widehat{C}_{.bc}^{a}}{\partial y^{d}}%
-\frac{\partial \widehat{C}_{.bd}^{a}}{\partial y^{c}}+\widehat{C}_{.bc}^{e}%
\widehat{C}_{.ed}^{a}-\widehat{C}_{.bd}^{e}\widehat{C}_{.ec}^{a}$
 from (\ref{dcurv}) and contracting the indices in order to obtain the Ricci
coefficients,%
 $\widehat{R}_{bc}=\frac{\partial \widehat{C}_{.bc}^{d}}{\partial y^{d}}-\frac{%
\partial \widehat{C}_{.bd}^{d}}{\partial y^{c}}+\widehat{C}_{.bc}^{e}%
\widehat{C}_{.ed}^{d}-\widehat{C}_{.bd}^{e}\widehat{C}_{.ec}^{d},$ we compute
\begin{equation*}
\widehat{R}_{bc}=\left( \widehat{C}_{.bc}^{4}\right) ^{\ast }-\partial _{c}%
\widehat{C}_{b}+\widehat{C}_{.bc}^{4}\widehat{C}_{4}-\widehat{C}_{.b4}^{4}%
\widehat{C}_{.4c}^{4}-\widehat{C}_{.b5}^{4}\widehat{C}_{.4c}^{5}-\widehat{C}%
_{.b4}^{5}\widehat{C}_{.5c}^{4}-\widehat{C}_{.b5}^{5}\widehat{C}_{.5c}^{5}.
\end{equation*}%
There are nontrivial values,
 $\widehat{R}_{44} = \left( \widehat{C}_{.44}^{4}\right) ^{\ast }-\widehat{C}%
_{4}^{\ast }+\widehat{C}_{44}^{4}(\widehat{C}_{4}-\widehat{C}%
_{44}^{4})-\left( \widehat{C}_{.45}^{5}\right) ^{2}$ and
 $\widehat{R}_{55} =\left( \widehat{C}_{.55}^{4}\right) ^{\ast }-\widehat{C}%
_{.55}^{4}\left( -\widehat{C}_{4}+2\widehat{C}_{.45}^{5}\right)$
 resulting in%
\begin{equation*}
\widehat{R}_{4}^{4}=\widehat{R}_{5}^{5}=\frac{1}{2h_{4}h_{5}}\left[
-h_{5}^{\ast \ast }+\frac{\left( h_{5}^{\ast }\right) ^{2}}{2h_{5}}+\frac{%
h_{4}^{\ast }h_{5}^{\ast }}{2h_{4}}\right]
\end{equation*}%
which is just (\ref{4ep2a}).

Computations for higher shells, with $k=1,2,...$ are similar.

Theorem \ref{th2} is proven.

\end{document}